\documentclass[fleqn,11pt,twoside]{article}

\usepackage{amsthm,amssymb,amsfonts,color,xcolor,epsfig,graphics,subfigure}

\usepackage[T1]{fontenc}
\usepackage[utf8]{inputenc}
\usepackage{amsmath,graphicx,latexsym,lscape,mathtools,mathrsfs,bm}
\usepackage[hyperfootnotes=false,hidelinks]{hyperref}
\makeatletter
\newcommand{\copyrightnote}[2]{{\renewcommand{\thefootnote}{}
 \footnotetext{\small\it
\begin{flushleft}
 \copyright \ #1   #2  
\end{flushleft}}}}

\newcommand{\Name}[1]{\begin{flushleft}
                       \LARGE \bf #1
                       \end{flushleft}\vspace{-3mm}}

\newcommand{\Author}[1]{\begin{flushleft}
                       \it #1 \end{flushleft}}

\newcommand{\Address}[1]{\begin{flushleft}
                       \it #1 \end{flushleft}}

\newcommand{\Date}[1]{\begin{flushleft}
                      \small  \it #1 \end{flushleft}}

\newcommand{\evenhead}{Author \ name}
\newcommand{\oddhead}{Article \ name}

\renewcommand{\@evenhead}{
\hspace*{-3pt}\raisebox{-15pt}[\headheight][0pt]{\vbox{\hbox to \textwidth
{\thepage \hfil \evenhead}\vskip4pt \hrule}}}
\renewcommand{\@oddhead}{
\hspace*{-3pt}\raisebox{-15pt}[\headheight][0pt]{\vbox{\hbox to \textwidth
{\oddhead \hfil \thepage}\vskip4pt\hrule}}}
\renewcommand{\@evenfoot}{}
\renewcommand{\@oddfoot}{}

\long\def\@makecaption#1#2{%
  \vskip\abovecaptionskip
  \sbox\@tempboxa{\small \textbf{#1.}\ \ #2}%
  \ifdim \wd\@tempboxa >\hsize
    {\small \textbf{#1.}\ \ #2}\par
  \else
    \global \@minipagefalse
    \hb@xt@\hsize{\hfil\box\@tempboxa\hfil}%
  \fi
  \vskip\belowcaptionskip}

\newcommand{\JNMPnumberwithin}[3][\arabic]{%
  \@ifundefined{c@#2}{\@nocounterr{#2}}{%
    \@ifundefined{c@#3}{\@nocnterr{#3}}{%
      \@addtoreset{#2}{#3}%
      \@xp\xdef\csname the#2\endcsname{%
        \@xp\@nx\csname the#3\endcsname .\@nx#1{#2}}}}%
}

\renewenvironment{proof}[1][\proofname]{\par
  \normalfont
  \topsep6\p@\@plus6\p@ \trivlist
  \item[\hskip\labelsep\textbf{%
    #1\@addpunct{.}}]\ignorespaces
}{%
  \qed\endtrivlist
}

\newcommand{\resetfootnoterule} {
  \renewcommand\footnoterule{%
  \kern-3\p@
  \hrule\@width.4\columnwidth
  \kern2.6\p@}
}

\renewcommand{\footnoterule}{}

\makeatother

\numberwithin{equation}{section}
\theoremstyle{plain}
\newtheorem{theorem}{Theorem}[section]
\newtheorem{lemma}[theorem]{Lemma}

\theoremstyle{definition}

\theoremstyle{remark}
\newtheorem{remark}[theorem]{Remark}

\newcommand{\id}{\mathrm{id}}
\newcommand{\ket}[1]{\left|#1\right\rangle}
\newcommand{\bra}[1]{\left\langle #1\right|}

\begin{document}

\renewcommand{\evenhead}{ {\LARGE\textcolor{blue!10!black!40!green}{{\sf \ \ \ ]ocnmp[}}}\strut\hfill 
P Padmanabhan, S Maity, A Sinha and V Korepin
}
\renewcommand{\oddhead}{ {\LARGE\textcolor{blue!10!black!40!green}{{\sf ]ocnmp[}}}\ \ \ \ \  
Symmetries of the Generalized Yang--Baxter Equations
%: Exact Solution via Painleve' Symmetry Reduction
}

%%%% Matter for the first page 
\thispagestyle{empty}
\newcommand{\FistPageHead}[3]{
\begin{flushleft}
\raisebox{8mm}[0pt][0pt]
{\footnotesize \sf
\parbox{150mm}{{\textcolor{blue!10!black!40!green}{{\bf Open Communications in Nonlinear Mathematical Physics}}}
\ \ {Special Issue: Hietarinta}, 2026\\[0.1cm]
\strut\hfill 
ocnmp:18633
pp #2\hfill {\sc #3}}}\vspace{-13mm}
\end{flushleft}}

\FistPageHead{1}{\pageref{firstpage}--\pageref{lastpage}}{ \ \ }

\strut\hfill

\strut\hfill

\copyrightnote{The authors.}{Distributed under a Creative Commons Attribution 4.0 International License.}

\begin{center}

{\bf {\large A Special OCNMP Issue in Honour of Jarmo Hietarinta}}\\[0.2cm]
{\bf {\large on the Occasion of his 80th Birthday}}
\end{center}

\smallskip

\Name{Symmetries of the Generalized Yang--Baxter Equations}

\Author{Pramod Padmanabhan, Somnath Maity, Akash Sinha and Vladimir Korepin}

\Address{P. Padmanabhan, S. Maity and A. Sinha: Department of Physics, School of Basic Sciences,\\
Indian Institute of Technology Bhubaneswar, 752050, India.\\
E-mail: pramod23phys@gmail.com; somnathmaity126@gmail.com; akash26121999@gmail.com}
\Address{V. Korepin: C. N. Yang Institute for Theoretical Physics, Stony Brook University,\\
New York 11794, USA. E-mail: vladimir.korepin@stonybrook.edu}

\Date{Received June 26, 2026; Accepted September 8, 2026}

\setcounter{equation}{0}

\smallskip

\noindent
{\bf Citation format for this Article:}\newline
Pramod Padmanabhan, Somnath Maity, Akash Sinha and Vladimir Korepin,
Symmetries of the Generalized Yang--Baxter Equations,
{\it Open Commun. Nonlinear Math. Phys.}, Special Issue:\,Hietarinta, ocnmp:18633, \pageref{firstpage}--\pageref{lastpage}, 2026.

\strut\hfill

\noindent
{\bf DOI:}
{\it 10.46298/ocnmp.18633}
\strut\hfill

\begin{abstract}

\noindent
The generalized Yang-Baxter equations are multi-site versions of the standard Yang-Baxter equation. When spectral parameters are included, such equations are expected to lead to integrable Hamiltonians with local interactions involving multiple degrees of freedom. In this work we characterize both the continuous and discrete symmetries of these equations required to establish an equivalence class of solutions. We find that the set of such symmetries depends on the number of sites on which the equation is supported. In several cases there are more symmetries than the standard Yang-Baxter equation, thus placing heavy constraints on the number of inequivalent solutions and the associated integrable models. As an application we show how to restrict the generic $16\times16$ ansatz to the 30-dimensional subspace invariant under a chosen discrete symmetry subgroup. An explicit $16\times 16$ solution is also written down.

\end{abstract}

\label{firstpage}

\section{Introduction}
\label{sec:intro}
%------------------------------------------------
One dimensional spin chains are the most well-studied examples of quantum integrable systems. These systems are typically governed by Hamiltonians with local interactions between nearest-neighbor spins. A large class of them can be systematically derived from certain solutions of the {\it Yang-Baxter equation} (YBE) \cite{YangCN1967,BAXTER1972193} using the formalism of the {\it quantum inverse scattering method} (QISM) \cite{Korepin1993QuantumIS,takhtadzhyan1979,Takhtadzhan_1979,Sklyanin1982QuantumVO,slavnov2019algebraicbetheansatz}. A natural generalization is to extend the study of quantum integrable systems to Hamiltonians with quasi-local interactions. The latter involve interactions between a finite number of consecutive spins going beyond the nearest-neighbor coupling terms. Some examples of Hamiltonians with multi-site interactions have appeared in the past as ladder integrable models and integrable models with staggered spectral parameters \cite{Ikhlef2009TheSV,Frahm1997PropertiesOT,Frahm_1999,Batchelor_2000,Ambjorn2000IntegrableLT,Frahm1995IntegrableMO,foerster2000integrablegeneralisedspinladder,Batchelor_2007}.  To obtain such Hamiltonians with multi-site interaction terms using the formalism of QISM, it is essential to look for a multi-site generalization of the YBE. An opportunity for this is presented in the so-called {\it generalized Yang-Baxter equation} (gYBE), whose constant version or the version without spectral parameters was introduced in \cite{rowell2010extraspecialtwogroupsgeneralizedyangbaxter}. The $(d,l,m)$-gYBE is an operator equation defined as follows:
\begin{equation}\label{eq:gYBE}
     R_{i_1\cdots i_l}R_{i_{m+1}\cdots i_{m+l}}R_{i_1\cdots i_l} = R_{i_{m+1}\cdots i_{m+l}} R_{i_1\cdots i_l}R_{i_{m+1}\cdots i_{m+l}}~;~m<l\in\mathbb{Z}^+. 
\end{equation}
The solution of this equation, known as the $R$-matrix, acts on $l$ sites. On each of these sites a local Hilbert space of dimension $d$ is attached. This is typically taken to be $\mathbb{C}^d$ in the context of spin chains. The integer $m$ denotes the shift in the indices between two consecutive $R$-matrices in the $(d,l,m)$-gYBE. Notice that in this notation, $l=2, m=1$ reduces \eqref{eq:gYBE} to the standard YBE. The index structure of the gYBE resembles the index structure of the {\it braid relations}. Thus, the constant solutions of the gYBE can also be used to obtain a representation of the full braid group. For this the translated \(R\)-operators must additionally satisfy the appropriate far-commutativity relations. For this reason we will sometimes use the term braid relations in place of the YBE or the gYBE in the rest of the text.  

These multi-site $R$-matrices satisfying the gYBE have been interpreted as operators generating quantum entanglement and their unitary versions have been used as quantum gates in quantum circuits \cite{chen2012generalized,Padmanabhan_2020,padmanabhan2020generating,rowell2010extraspecialtwogroupsgeneralizedyangbaxter}. This interpretation develops an idea relating topological and quantum entanglement \cite{Aravind1997,sugita2007borromeanentanglementrevisited,padmanabhan2019quantum,kauffman2019topological,quinta2018classifying,balasubramanian2017multi,balasubramanian2018entanglement,dwivedi2018entanglement,melnikov2019topological,rowell2018mathematics}. For applications of \eqref{eq:gYBE} to quantum integrability, we need to introduce complex spectral parameters into the gYBE. However, for the purposes of this work we will see that introducing these parameters does not lead to drastically new results and so we will drop their usage in the rest of this work.

While many solutions of the constant gYBE have been found, the broader goal of classifying solutions, say along the lines of Jarmo Hietarinta's $4\times 4$ solutions of the YBE \cite{HIETARINTA-PLA,hietarinta1993-JMP-Long,Hietarinta1993-BookChapter,hietarinta-conferenceProceeding,MSPK-Hietarinta}, is still lacking for the gYBE. This is a difficult problem using numerics as the dimensions of the $R$-matrices drastically increase with $l$ and $m$. This follows from the following properties of the $(d,l,m)$-gYBE. First we note that the $(d,l,m)$-gYBE is non-trivially supported on the space $\bigotimes_{j=1}^{l+m}\mathbb{C}_j^d$ and the $R$-matrix is of size $d^l\times d^l$. This implies that the $(d,l,m)$-gYBE is a set of $d^{2(l+m)}$ equations in $d^{2l}$ variables. This makes the multi-site $R$-matrices solutions of a highly non-linear and overdetermined system of equations. Therefore any reasonable attempt at a classification will require clever ans\"{a}tze that simplify the choices of these multi-site $R$-matrices for a given $d$. Such simplifications can be obtained by studying the symmetries of the $(d,l,m)$-gYBE. This is expected to reduce the number of unknown matrix elements of the $R$-matrices from $d^{2l}$ to a more manageable number by enforcing these symmetries on their ans\"atze. The goal of this work is precisely to characterize such symmetries and contrast them with those seen in the standard YBE or the $l=2, m=1$ case. Thus we hope that the results in this work will lay the foundation for classifying solutions of the gYBE.

The results are organized as follows. Sections \ref{sec:c-symmetries} and \ref{sec:d-symmetries} discuss the continuous and discrete symmetries of the gYBE. A concrete application of these methods is considered in Section \ref{sec:solutions}, where we explicitly show how to simplify the ansatz of a $(2,4,2)$-gYBE from a matrix with 256 entries to one with just 30 entries. Some future directions make up Section \ref{sec:conclusion}.

It should be noted that the symmetries found in this work are obtained by generalizing known symmetries of the standard YBE to the multi-site setting. This generalization is already non-trivial and produces several interesting features of the gYBE. There could be other symmetries that have no counterpart in the standard YBE setting, but that will not be the focus of the current work. Thus we do not claim to have exhausted all the symmetries of the gYBE.

%------------------------------------------------
\section{Continuous Symmetries}
\label{sec:c-symmetries}
%------------------------------------------------
We will find the set of continuous symmetries of the $(d,l,m)$-gYBE by taking cues from analogous symmetries of the standard YBE. In this and the following sections we will always consider the local Hilbert space to be the $d$-dimensional complex vector space $\mathbb{C}^d$. Before going further we clarify the notion of symmetries of the gYBE. Let $R_{ij}$ be a solution of the $(d,l,m)$-gYBE. Then if a transformed version of $R_{ij}$ also solves the gYBE, this transformation is called a symmetry of the gYBE. The continuous transformations come in at least four types. Two of them are {\it unconstrained continuous symmetries}:
\begin{enumerate}
    \item Scaling symmetry.
    \item Homogeneous, $\mathrm{GL}(d,\mathbb{C})$ gauge symmetries.
\end{enumerate}
The other two are {\it constrained continuous symmetries}:
\begin{enumerate}
    \item Symmetries generated by certain invertible elements of the underlying representation space, subject to appropriate compatibility conditions with $R_{ij}$, also known as {\it twist symmetries}. 
    \item Constrained symmetries generated by certain invertible braid operators, also known as {\it generalized twist symmetries}.
\end{enumerate}
The reason for the terminology will become clear as we discuss each of these symmetries and their counterparts for the gYBE separately. In the rest of the section we will use $R$ $\left(\tilde{R}\right)$ for the untransformed (transformed) YBE or gYBE solutions.

%------------------------------------------------
\subsection{Unconstrained Continuous Symmetries}
\label{subsec:unconstrained-c-symmetries}
%------------------------------------------------
The scaling symmetry of the YBE is the simplest among the four continuous symmetries. This is given by 
\begin{equation}\label{eq:scaling-symmetry-YBE}
   R_{ij} \rightarrow \tilde{R}_{ij} = \kappa~R_{ij}~;~\kappa\in\mathbb{C}^\times.
\end{equation}
Substituting $\tilde{R}$ into the YBE multiplies both sides of the equation by the same constant factor $\kappa^3$, making the scaling transformation a symmetry of the YBE. This easily generalizes to the gYBE case in a straightforward manner as the transformed multi-site $R$,
\begin{equation}\label{eq:scaling-symmetry-gYBE}
   R_{j_1\cdots j_l} \rightarrow \tilde{R}_{j_1\cdots j_l} = \kappa~R_{j_1\cdots j_l}~;~\kappa\in\mathbb{C}^\times,
\end{equation}
also satisfies the $(d,l,m)$-gYBE.

Next, we consider the following similarity transform by an invertible operator $Q$\footnote{The operators $Q_i$ and $Q_j$ are translated versions of the same operator $Q$. This is taken to be the case for all the continuous symmetries discussed in this work.},
\begin{equation}\label{eq:gldc-gauge-symmetries-YBE}
  R_{ij}\rightarrow\tilde{R}_{ij} = Q_iQ_j~R_{ij}~Q_i^{-1}Q_j^{-1}.
\end{equation}
The proof that this is a symmetry is immediate by substituting $\tilde{R}$ into the appropriate braid relation obtained from \eqref{eq:gYBE}. The invertible operators $Q$ belong to $\mathrm{GL}(d,\mathbb{C})$. They can depend on a set of parameters that are allowed to vary over $\mathbb{C}$, making these symmetries a set of continuous symmetries akin to the infinite set of {\it gauge symmetries}\footnote{Strictly speaking the terminology, gauge symmetries, should be reserved for quantities depending on a continuous parameter, like the spacetime coordinate. In the context of the YBE this terminology is more suited for spectral parameter dependent YBE solutions. Regardless, we will adopt this terminology for the constant YBE solutions as well.}. Note that the same operator $Q$ has to act on both the indices $i$ and $j$. They also need to be factorized over the tensor product space $\mathbb{C}^d\otimes\mathbb{C}^d$. For these reasons we call these symmetries homogeneous $\mathrm{GL}(d,\mathbb{C})$ gauge symmetries. Relaxing these conditions will lead us to {\it entangled operators} such as $Q_{ij}$ as the similarity transform. But they need to satisfy additional constraints for it to preserve the YBE structure. Some such constraints will be studied in continuous symmetries that we will soon consider. 

In the multi-site setting, it is easily seen that the fully factorized local similarity transforms, generated by $Q\in\mathrm{GL}(d,\mathbb{C})$
\begin{equation}\label{eq:gldc-gauge-symmetry-gYBE}
    R_{j_1\cdots j_l}\rightarrow \left(Q_{j_1}\cdots Q_{j_l}\right)R_{j_1\cdots j_l}\left(Q_{j_1}^{-1}\cdots Q_{j_l}^{-1}\right),
\end{equation}
also satisfies the gYBE in \eqref{eq:gYBE}. However, in this case we have more such continuous symmetries that are not of the fully factorized form. This enhancement in the number of symmetries comes from the possibility of allowing entangled operators as valid similarity transforms. This depends on the values of $l$ and $m$. The reason for this increase is attributed to the enlarged support of the $R$-matrices. The precise form of these additional symmetries can be stated in the following Theorem.
\begin{theorem}[Enhanced homogeneous $\mathrm{GL}(d,\mathbb{C})$ gauge symmetries]\label{thm:gauge-symmetries}
Let $R$ be a solution of the $(d,l,m)$-gYBE, \eqref{eq:gYBE}. Then when $r=\gcd(l,m)>1$, the gYBE has an additional set of gauge symmetries given by the transformation
\begin{equation*}
    R_{j_1\cdots j_l}\rightarrow \left(\prod\limits_{n=1}^{\frac{l}{r}}~Q_{j_{(n-1)r+1}\cdots j_{nr}}\right)R_{j_1\cdots j_l}\left(\prod\limits_{n=1}^{\frac{l}{r}}~Q_{j_{(n-1)r+1}\cdots j_{nr}}^{-1}\right).
\end{equation*}
Each of the $Q$ factors is an entangled operator on $\bigotimes_{k=1}^r~\mathbb{C}^d_{j_k}$. Thus, $Q\in\mathrm{GL}(d^r,\mathbb{C})$. There are $\frac{l}{r}$ such factors supported on $\bigotimes_{k=1}^l~\mathbb{C}^d_{j_k}$.
\end{theorem}
\begin{proof}
    The proof follows from a direct substitution of the transformed $R$-matrix into the $(d,l,m)$-gYBE in \eqref{eq:gYBE}. We find that both sides of the equation are conjugated by the factor 
    \begin{equation*}
        \left(\prod\limits_{n=1}^{\frac{l+m}{r}}~Q_{j_{(n-1)r+1}\cdots j_{nr}}\right),
    \end{equation*}
    rendering the equality of the transformed gYBE. This completes the proof.
\end{proof}

The fully factorized gauge symmetry in \eqref{eq:gldc-gauge-symmetry-gYBE} can be viewed as a special case of the entangled gauge symmetries of Theorem \ref{thm:gauge-symmetries}. Thus when $r>1$, the possibility of such entangled gauge symmetries is more constraining for the possible inequivalent families of gYBE solutions when compared to the gauge symmetries of the standard YBE.

Now we provide a couple of examples to illustrate Theorem \ref{thm:gauge-symmetries}. Consider the $(d,6,2)$ and the $(d,6,4)$-gYBEs. In both cases $r=2$ and we find that the transformation,
\begin{equation*}
    R_{j_1\cdots j_6}\rightarrow \left(Q_{j_1j_2}Q_{j_3j_4}Q_{j_5j_6}\right)R_{j_1\cdots j_6}\left(Q_{j_1j_2}^{-1}Q_{j_3j_4}^{-1}Q_{j_5j_6}^{-1}\right),
\end{equation*}
leaves them invariant. Next, consider the $(d,6,3)$-gYBE. In this case $r=3$ and the additional entangled transformation is given by
\begin{equation*}
    R_{j_1\cdots j_6}\rightarrow \left(Q_{j_1j_2j_3}Q_{j_4j_5j_6}\right)R_{j_1\cdots j_6}\left(Q_{j_1j_2j_3}^{-1}Q_{j_4j_5j_6}^{-1}\right).
\end{equation*}

We can now summarize the set of continuous symmetries for the $(d,l,m)$-gYBE, including both the scaling and the $\mathrm{GL}(d^r,\mathbb{C})$ gauge symmetries, into the following compact expressions:
\begin{multline}\label{eq:unconstrianed-c-symmetries-gYBE}
    R_{j_1\cdots j_l} \rightarrow\begin{cases}
        \kappa~\left(Q_{j_1}\cdots Q_{j_l}\right)R_{j_1\cdots j_l}\left(Q_{j_1}^{-1}\cdots Q_{j_l}^{-1}\right),~~r=1 \\
        \kappa\left(\prod\limits_{n=1}^{\frac{l}{r}}Q_{j_{(n-1)r+1}\cdots j_{nr}}\right)R_{j_1\cdots j_l}\left(\prod\limits_{n=1}^{\frac{l}{r}}Q_{j_{(n-1)r+1}\cdots j_{nr}}^{-1}\right),~~r>1,
    \end{cases}
\end{multline}
where $r=\gcd(l,m)$ and $\kappa\in\mathbb{C}^\times$.

%------------------------------------------------
\subsection{Constrained Continuous Symmetries}
\label{subsec:constrained-c-symmetries}
%------------------------------------------------
The next set of continuous symmetries are generated by invertible operators. We will state and prove these in the following Lemmas. 
\begin{lemma}[Twist transformations or the constrained $\mathrm{GL}(d,\mathbb{C})$ gauge symmetries of the YBE - I]\label{lem:twist-YBE}
Consider invertible operators $V$ and $W$. If $\left [R_{ij}, V_iW_j \right]=0$, then both $V_iR_{ij}V_i^{-1}$ and $W_jR_{ij}W_j^{-1}$, satisfy the YBE.  
\end{lemma}
\begin{proof}
    Using $V_iR_{ij}V_i^{-1}=W_j^{-1}R_{ij}W_j$, the l.h.s of the YBE for the transformed $R$ is given by
    \begin{align*}
        \tilde{R}_{ij}\tilde{R}_{jk}\tilde{R}_{ij} & = V_iR_{ij}V_i^{-1}V_jR_{jk}V_j^{-1}V_iR_{ij}V_i^{-1} \\ 
        & = V_iR_{ij}V_i^{-1}W_k^{-1}R_{jk}W_kV_iR_{ij}V_i^{-1} \\
        & = V_iW_k^{-1}\left(R_{ij}R_{jk}R_{ij}\right)W_kV_i^{-1}.
    \end{align*}
    The r.h.s of the YBE for the transformed $R$ becomes,
    \begin{align*}
        \tilde{R}_{jk}\tilde{R}_{ij}\tilde{R}_{jk} & = V_jR_{jk}V_j^{-1}V_iR_{ij}V_i^{-1}V_jR_{jk}V_j^{-1} \\
        & = W_k^{-1}R_{jk}W_kV_iR_{ij}V_i^{-1}W_k^{-1}R_{jk}W_k \\
        & = V_iW_k^{-1}\left(R_{jk}R_{ij}R_{jk}\right)W_kV_i^{-1}.
    \end{align*}
    In a similar manner, $W_jR_{ij}W_j^{-1}$, can be shown to satisfy the YBE. This completes the proof.
\end{proof}
Special instances of the symmetry in Lemma \ref{lem:twist-YBE} have appeared in the Fermi-Hubbard model setup as shown in Chapter 12, Lemma 11 of \cite{esslerfrahmgohmannklumperkorepin2005}. 

\begin{lemma}[Twist transformations or the constrained $\mathrm{GL}(d,\mathbb{C})$ gauge symmetries of the YBE - II]\label{lem:twist-YBE-2}
    Consider invertible operators $V$, $W$, $U$ and $T$ belonging to $\mathrm{GL}(d,\mathbb{C})$. If $\left[R_{ij},V_iU_j \right]=\left[R_{ij},T_iW_j \right]=0$, then the transformation $\tilde{R}_{ij}=Q_{ij}R_{ij}Q_{ij}^{-1}$ satisfies the YBE, when 
    $$ Q_{ij} \in \left\{V_iW_j, T_iU_j, V_iT_j, U_iW_j, T_iV_j, W_iU_j   \right\}.  $$
\end{lemma}
\begin{proof} 
We will write down the proofs for two of the $Q_{ij}$ operators. They will serve as the prototypes for the remaining cases. The steps for each case are similar to those used to prove Lemma \ref{lem:twist-YBE}. Consider the case when $Q_{ij}=V_iW_j$. The left hand side then simplifies as 
\begin{align*}
    \tilde{R}_{ij}\tilde{R}_{jk}\tilde{R}_{ij} & =
V_iW_k\left(W_jR_{ij}W_j^{-1}\right)V_jR_{jk}V_j^{-1}\left(W_jR_{ij}W_j^{-1}\right)V_i^{-1}W_k^{-1} \\
    & = V_iT_i^{-1}W_kR_{ij}U_k^{-1}R_{jk}U_k R_{ij}T_iV_i^{-1}W_k^{-1} \\
    & = V_iT_i^{-1}W_kU_k^{-1}\left(R_{ij}R_{jk}R_{ij}\right)T_iV_i^{-1}U_kW_k^{-1}.
\end{align*}
In a similar manner we can show that 
\begin{align*}
    \tilde{R}_{jk}\tilde{R}_{ij}\tilde{R}_{jk} = V_iT_i^{-1}W_kU_k^{-1}\left(R_{jk}R_{ij}R_{jk}\right)T_iV_i^{-1}U_kW_k^{-1}.
\end{align*}
A similar type of proof can be used for $Q_{ij}=T_iU_j$. In this case the transformed YBE is conjugated by $T_iV_i^{-1}U_kW_k^{-1}$.

Next we consider $Q_{ij}=V_iT_j$. Now the left hand side of the YBE simplifies as 
\begin{align*}
    \tilde{R}_{ij}\tilde{R}_{jk}\tilde{R}_{ij} & = V_iT_jT_kU_k^{-1} R_{ij}\left(T_j^{-1}R_{jk}T_j\right)R_{ij}V_i^{-1}T_j^{-1}U_kT_k^{-1} \\
    & = V_iT_jT_kU_k^{-1}W_k\left(R_{ij}R_{jk}R_{ij}\right)V_i^{-1}T_j^{-1}W_k^{-1}U_kT_k^{-1},
\end{align*}
and the right hand side becomes
\begin{align*}
   \tilde{R}_{jk}\tilde{R}_{ij}\tilde{R}_{jk} & = V_jT_kR_{jk}V_j^{-1}V_iT_jR_{ij}V_i^{-1}T_j^{-1}V_jR_{jk}V_j^{-1}T_k^{-1} \\
   & = V_iT_kU_k^{-1}R_{jk}\left(T_jR_{ij}T_j^{-1}\right)R_{jk}V_i^{-1}U_kT_k^{-1} \\
   & =  
V_iT_jT_kU_k^{-1}\left(T_j^{-1}R_{jk}T_j\right)R_{ij}\left(T_j^{-1}R_{jk}T_j\right)V_i^{-1}T_j^{-1}U_kT_k^{-1} \\
& = V_iT_jT_kU_k^{-1}W_k\left(R_{jk}R_{ij}R_{jk}\right)V_i^{-1}T_j^{-1}W_k^{-1}U_kT_k^{-1}.
\end{align*}
A similar type of proof can be used for the remaining $Q_{ij}$'s, namely $U_iW_j$, $T_iV_j$ and $W_iU_j$. The conjugation factors for the YBE in each of these cases is given by $U_iT_i^{-1}V_iU_jW_k$, $T_iV_jV_kW_k^{-1}U_k$ and $W_iV_i^{-1}T_iW_jU_k$, respectively.
\end{proof}

Unlike the $Q$ operators of the homogeneous gauge symmetries, the invertible operators $V$, $W$, $T$ and $U$ are constrained by commutation relations with $R_{ij}$'s. As their entries can vary over $\mathbb{C}$, they belong to $\mathrm{GL}(d,\mathbb{C})$ and thus we will denote these continuous symmetries as {\it constrained $\mathrm{GL}(d,\mathbb{C})$ gauge symmetries}. Note also that these symmetries are not special cases of the homogeneous gauge symmetries of \eqref{eq:gldc-gauge-symmetries-YBE}, as the symmetries in Lemma \ref{lem:twist-YBE} rotate the YBE solutions by acting on only one of the two sites. In this sense it is a partial similarity transform when compared to the full similarity transforms of \eqref{eq:unconstrianed-c-symmetries-gYBE}. On the other hand, the constrained symmetries of Lemma \ref{lem:twist-YBE-2}, reduce to the unconstrained symmetries in \eqref{eq:gldc-gauge-symmetries-YBE} only when $V=W$ or $U=T$.  
Lemma 2.3 provides a generalization of the twist transformation found in \cite{de2021yang, Kundu_2001}.

The analogous set of constrained gauge symmetries or twist transformations for the $(d,l,m)$-gYBE case, depends on the values of $l$ and $m$. For both the odd $(l=2p+1)$, and even $(l=2p)$ cases, the set of symmetry transformations differ for two sets of $m$ values. In the odd case they differ in the sets $m\in\{1,\cdots, p\}$ and $m\in\{p+1,\cdots, 2p\}$, whereas in the even case these sets are given by $m\in\{1,\cdots, p-1\}$ and $m\in\{p,\cdots, 2p-1\}$. 

The following two Theorems account for the twist symmetries analogous to those formulated in Lemma \ref{lem:twist-YBE} in these two sets, for the gYBE case.
\begin{theorem}[Twist transformations or the constrained $\mathrm{GL}(d,\mathbb{C})$ gauge symmetries of the gYBE - I - 1]\label{thm:twist-gYBE-1}
    For $l=2p+1(l=2p)$, choose the value of $m$ from the sets $\{p+1,\cdots, 2p\}$ $\left(\{p,\cdots, 2p-1\}\right)$. Use this to fix the value of $r=l-m>0$. Now consider the set of invertible operators $V_{i_1\cdots i_s}$ and $W_{i_{s+1}\cdots i_l}$ for $s\in \{r,\cdots, l-r\}$. If $$\left[R_{i_1\cdots i_l}, V_{i_1\cdots i_s}W_{i_{s+1}\cdots i_l} \right] = 0, $$ and $R_{i_1\cdots i_l}$ satisfies the $(d,l,m)$-gYBE, then both
\[
V_{i_1\cdots i_s}R_{i_1\cdots i_l}V_{i_1\cdots i_s}^{-1}
\quad\text{and}\quad
W_{i_{s+1}\cdots i_l}R_{i_1\cdots i_l}W_{i_{s+1}\cdots i_l}^{-1}
\]
satisfy the $(d,l,m)$-gYBE. 
\end{theorem}
\begin{proof}
    The proofs for both the transformed operators are similar to that shown in Lemma \ref{lem:twist-YBE}. This is immediately seen with the following replacements:
    \begin{enumerate}
        \item $V_i\rightarrow V_{i_1\cdots i_s} $, $V_j\rightarrow V_{i_{1+m}\cdots i_{s+m}}$, $W_k\rightarrow W_{i_{s+m+1}\cdots i_{l+m}}$.
        \item $R_{ij}\rightarrow R_{i_1\cdots i_l}$, ~ $R_{jk}\rightarrow R_{i_{1+m}\cdots i_{l+m}}$.
    \end{enumerate} 
    \end{proof}
Note that the two symmetries $V$ and $W$ have support on multiple sites and are thus entangled operators, in contrast to the analogous YBE symmetries in Lemma \ref{lem:twist-YBE}. For a given value of $s$, they belong to $\mathrm{GL}(d^s, \mathbb{C})$ and $\mathrm{GL}(d^{l-s}, \mathbb{C})$ respectively.  

The Theorem \ref{thm:twist-gYBE-1} shows that for every $s$ value there are two twist symmetries generated by the entangled invertible operators $V$ and $W$. Thus for a given $l$ and a $m$, chosen from the appropriate sets of values, there are a total of $$2(l-2r+1)=2(2m-l+1),$$ such constrained gauge symmetries or twist transformations. As $m$ can take $p$ values, for both odd and even $l$, they account for the twist symmetries of the gYBEs in these cases. 

A special case occurs for even $l=2p$ and $m=r=p$, corresponding to the $(d,2p,p)$-gYBE. In this case there are just two twist symmetries given by $V_{i_1\cdots i_p}$ and $W_{i_{p+1}\cdots i_{2p}}$. This is precisely the same number of symmetries that appear for the standard YBE case as shown in Lemma \ref{lem:twist-YBE}. This is expected as the $(d,2p,p)$-gYBE, is closest to the standard YBE in terms of the index structure of the equation. It can be obtained from the YBE simply by enlarging each of the indices appearing in the YBE by an equal amount, $p$. Thus we expect several properties of the $(d,2p,p)$-gYBE, especially its symmetries, to be similar to that of the standard YBE.

A simple example is useful to illustrate the structure of the symmetry operators of Theorem \ref{thm:twist-gYBE-1}. Consider the $l=5$ case. Then the appropriate values for $m$ are, $m=3,4$. Then the different twist transformations are given by:
\begin{align*}
     m=4 & ~:~ V_{i_1i_2i_3i_4}, W_{i_5}; V_{i_1i_2i_3}, W_{i_4i_5}; V_{i_1i_2}, W_{i_3i_4i_5}; V_{i_1},W_{i_2i_3i_4i_5} \\
     m=3 & ~:~ V_{i_1i_2i_3}, W_{i_4i_5}; V_{i_1i_2},W_{i_3i_4i_5}.
\end{align*}
The following Theorem accounts for the remaining sets of $m$ values in both the odd and even case.
\begin{theorem}[Twist transformations or the constrained $\mathrm{GL}(d,\mathbb{C})$ gauge symmetries of the gYBE - I - 2]\label{thm:twist-gYBE-2}
    For odd \(l=2p+1\) (even \(l=2p\)), choose $m$ from the set $\{1,\cdots, p\}$ $\left(\{1,\cdots, p-1\}\right)$. Then if $R_{i_1\cdots i_l}$ satisfies the $(d,l,m)$-gYBE and if there exist operators $V_{i_1\cdots i_m}$ and $W_{i_{l-m+1}\cdots i_l}$, that belong to $\mathrm{GL}(d^m,\mathbb{C})$ and satisfy 
    $$ \left[R_{i_1\cdots i_l}, V_{i_1\cdots i_m}W_{i_{l-m+1}\cdots i_l} \right] =0, $$
    then the transformations $$V_{i_1\cdots i_m}R_{i_1\cdots i_l}V_{i_1\cdots i_m}^{-1}~;~W_{i_{l-m+1}\cdots i_l}R_{i_1\cdots i_l}W_{i_{l-m+1}\cdots i_l}^{-1},$$ also satisfy the $(d,l,m)$-gYBE. 
\end{theorem}
\begin{proof}
The proof follows from a direct computation. Consider the l.h.s of the gYBE with $R$ transformed by $V$:
    \begin{align*}
        & \tilde{R}_{i_1\cdots i_l}\tilde{R}_{i_{1+m}\cdots i_{l+m}}\tilde{R}_{i_1\cdots i_l} \\
        = & V_{i_1\cdots i_m}R_{i_1\cdots i_l}V_{i_1\cdots i_m}^{-1} V_{i_{1+m}\cdots i_{2m}}R_{i_{1+m}\cdots i_{l+m}}V_{i_{1+m}\cdots i_{2m}}^{-1}V_{i_1\cdots i_m}R_{i_1\cdots i_l}V_{i_1\cdots i_m}^{-1} \\
        = & V_{i_1\cdots i_m}R_{i_1\cdots i_l}V_{i_1\cdots i_m}^{-1} W_{i_{l+1}\cdots i_{l+m}}^{-1}R_{i_{1+m}\cdots i_{l+m}}W_{i_{l+1}\cdots i_{l+m}}V_{i_1\cdots i_m}R_{i_1\cdots i_l}V_{i_1\cdots i_m}^{-1} \\
        = & V_{i_1\cdots i_m}W_{i_{l+1}\cdots i_{l+m}}^{-1}\left(R_{i_1\cdots i_l}R_{i_{m+1}\cdots i_{l+m}} R_{i_1\cdots i_l}\right)W_{i_{l+1}\cdots i_{l+m}} V_{i_1\cdots i_m}^{-1}.
    \end{align*}
    The r.h.s of the gYBE can similarly be simplified to 
    \begin{align*}
       &  \tilde{R}_{i_{1+m}\cdots i_{l+m}}\tilde{R}_{i_1\cdots i_l}\tilde{R}_{i_{1+m}\cdots i_{l+m}}  \\
      % = & V_{i_{1+m}\cdots i_{2m}}R_{i_{1+m}\cdots i_{l+m}}V_{i_{1+m}\cdots i_{2m}}^{-1}V_{i_1\cdots i_m}R_{i_1\cdots i_l}V_{i_1\cdots i_m}^{-1}V_{i_{1+m}\cdots i_{2m}}R_{i_{1+m}\cdots i_{l+m}}V_{i_{1+m}\cdots i_{2m}}^{-1} \\
      % = & W_{i_{l+1}\cdots i_{l+m}}^{-1}R_{i_{1+m}\cdots i_{l+m}}W_{i_{l+1}\cdots i_{l+m}}V_{i_1\cdots i_m}R_{i_1\cdots i_l}V_{i_1\cdots i_m}^{-1}W_{i_{l+1}\cdots i_{l+m}}^{-1}R_{i_{1+m}\cdots i_{l+m}}W_{i_{l+1}\cdots i_{l+m}} \\
       = & V_{i_1\cdots i_m}W_{i_{l+1}\cdots i_{l+m}}^{-1}\left(R_{i_{1+m}\cdots i_{l+m}}R_{i_{1}\cdots i_{l}} R_{i_{1+m}\cdots i_{l+m}}\right)W_{i_{l+1}\cdots i_{l+m}} V_{i_1\cdots i_m}^{-1},
    \end{align*}
    thus proving that the transformed $R$ satisfies the $(d,l,m)$-gYBE.
    A similar set of computations go through for the $R$ transformed by $W$ as well. 
\end{proof}
Thus for the gYBEs with these values of $m$, there are just two twist symmetries for any $m$, in contrast to those discussed in Theorem \ref{thm:twist-gYBE-1}. Consider two simple examples that clarify this Theorem. For $l=5$, we have $m=1,2$. The twist symmetries are then given by $V_{i_1}$, $W_{i_5}$ and $V_{i_1i_2}$, $W_{i_4i_5}$ respectively. Consider the even $l$ example of $l=6$, with $m=1,2$. The twist symmetries are then given by $V_{i_1}$, $W_{i_6}$ and $V_{i_1i_2}$, $W_{i_5i_6}$, respectively. 

Next we identify the twist symmetries analogous to those constructed in Lemma \ref{lem:twist-YBE-2} for the gYBE case. These symmetries will depend on the two sets of values of $m$ as elaborated above and thus can be formulated in the following Theorems.
%\begin{lemma}\label{lem:twist-gYBE-3} \textcolor{blue}{It is included as a special case of Theorem 2.7.}
 %    For $l$ odd (even), choose the value of $m$ from the sets $\{p+1,\cdots, 2p\}$ $\left(\{p,\cdots, 2p-1\}\right)$. Let $r=l-m>0$. Consider invertible and entangled operators $V_{i_1\cdots i_r}$ and $W_{i_{l-r+1}\cdots i_l}$ satisfying the conditions
  %   \begin{align*}
   %      \left [R_{i_1\cdots i_l}, V_{i_1\cdots i_r}V_{i_{l-r+1}\cdots i_l} \right] =0~;~\left [R_{i_1\cdots i_l}, W_{i_1\cdots i_r}W_{i_{l-r+1}\cdots i_l} \right] =0.
    % \end{align*}
    % Then if $R$ satisfies the $(d,l,m)$-gYBE, the transformed $R$-matrix
    % \begin{align*}
     %    \tilde{R}_{i_1\cdots i_l} = V_{i_1\cdots i_r}W_{i_{l-r+1}\cdots i_l}R_{i_1\cdots i_l}V^{-1}_{i_1\cdots i_r}W^{-1}_{i_{l-r+1}\cdots i_l},
     %\end{align*}
     %also satisfies the $(d,l,m)$-gYBE.
%\end{lemma}
%\begin{proof}
 %   The result immediately follows from noticing that for each of the $(d,l,m)$-gYBEs with $m$ chosen from these two sets, a $(d,2r,r)$-gYBE is embedded in the first and last $r=l-m$ indices of $R$. Then applying Lemma \ref{lem:twist-YBE-2} on these indices, with the entangled $V$ and $W$ operators, gives us the desired result. 
%\end{proof}
%$V^{(\chi)}_{i_1\cdots i_{\chi}}$, $T^{(\chi)}_{i_1\cdots i_{\chi}}$, $W^{(\chi)}_{i_{1+\chi}\cdots i_l}$, $U^{(\chi)}_{i_{1+\chi}\cdots i_l}$

\begin{theorem}[Twist transformations or the constrained $\mathrm{GL}(d,\mathbb{C})$ gauge symmetries of the gYBE - II - 1]\label{thm:twist-gYBE-4}
     For $l$ odd (even), choose the value of $m$ from the sets $\{p+1,\cdots, 2p\}$ $\left(\{p,\cdots, 2p-1\}\right)$ and let $r=l-m$. For $\chi\in \{m,r\}$, let $V^{(\chi)}$, $T^{(\chi)}$, $W^{(\chi)}$, $U^{(\chi)}$ be invertible entangled operators acting on the indicated tensor blocks. If $R_{i_1\cdots i_l}$ satisfies the $(d,l,m)$-gYBE and if the commutation relations
\begin{align*}
\left[
R_{i_1\cdots i_l}, V^{(\chi)}_{i_1\cdots i_{\chi}}U^{(\chi)}_{i_{1+\chi}\cdots i_l}
\right]= \left[ R_{i_1\cdots i_l},T^{(\chi)}_{i_1\cdots i_{\chi}}W^{(\chi)}_{i_{1+\chi}\cdots i_l}
\right] & =0, \qquad \chi\in \{m,r\}
%\\
%\left[
%R_{i_1\cdots i_l}, V_{i_1\cdots i_r}U_{i_{1+r}\cdots i_l}
%\right]= \left[ R_{i_1\cdots i_l},T_{i_1\cdots i_r}W_{i_{1+r}\cdots i_l}
%\right] & =0,
\end{align*}
hold, then, for each $\chi \in\{m,r\}$, the transformed $R$-matrices
\begin{align*}
\tilde{R}^{(\chi,1)}_{i_1\cdots i_l} & = 
V^{(\chi)}_{i_1\cdots i_{\chi}} W^{(\chi)}_{i_{1+\chi}\cdots i_l} R_{i_1\cdots i_l}
\left(V^{(\chi)}_{i_1\cdots i_{\chi}}\right)^{-1} \left(W^{(\chi)}_{i_{1+\chi}\cdots i_l}\right)^{-1}, \\
\tilde{R}^{(\chi,2)}_{i_1\cdots i_l} & = 
T^{(\chi)}_{i_1\cdots i_{\chi}} U^{(\chi)}_{i_{1+\chi}\cdots i_l} R_{i_1\cdots i_l}
\left(T^{(\chi)}_{i_1\cdots i_{\chi}}\right)^{-1} \left(U^{(\chi)}_{i_{1+\chi}\cdots i_l}\right)^{-1}
\end{align*}
are also solutions of the $(d,l,m)$-generalized Yang-Baxter equation. Here, $\chi=m$ and $\chi=r=(l-m)$ correspond respectively to the two different tensor block decompositions $m+(l-m)$ and $r+(l-r)$. This implies that the operators used for the $m$- and $r$-block decompositions are not being identified with one another.
\end{theorem}
\begin{proof}
For simplicity, we give the proof only for the case $\chi=m$. The other cases follow in the same way. Throughout this proof, we use the following simplified notation:
\begin{equation*}
V^{(m)}=V,\quad T^{(m)}=T,\quad W^{(m)}=W,\quad U^{(m)}=U.
\end{equation*}
We will begin with the l.h.s. of the transformed gYBE, which can be simplified through the following steps
\begin{align*}
    & \tilde{R}^{(m,1)}_{i_1\cdots i_l}\tilde{R}^{(m,1)}_{i_{1+m}\cdots i_{l+m}}\tilde{R}^{(m,1)}_{i_1\cdots i_l} \\
    = & V_{i_1\cdots i_m} W_{i_{1+m}\cdots i_l} R_{i_1\cdots i_l}
V_{i_1\cdots i_m}^{-1} W_{i_{1+m}\cdots i_l}^{-1} V_{i_{1+m}\cdots i_{2m}} W_{i_{1+2m}\cdots i_{l+m}} R_{i_{1+m}\cdots i_{l+m}}\\& \qquad
V_{i_{1+m}\cdots i_{2m}}^{-1} W_{i_{1+2m}\cdots i_{l+m}}^{-1}  V_{i_1\cdots i_m} W_{i_{1+m}\cdots i_l} R_{i_1\cdots i_l}
V_{i_1\cdots i_m}^{-1} W_{i_{1+m}\cdots i_l}^{-1}\\
= & V_{i_1\cdots i_m} T_{i_1\cdots i_m}^{-1} R_{i_1\cdots i_l} T_{i_1\cdots i_m} W_{i_{1+2m}\cdots i_{l+m}} U_{i_{1+2m}\cdots i_{l+m}}^{-1} R_{i_{1+m}\cdots i_{l+m}}\\& \qquad
U_{i_{1+2m}\cdots i_{l+m}}  W_{i_{1+2m}\cdots i_{l+m}}^{-1} T_{i_1\cdots i_m}^{-1}R_{i_1\cdots i_l}T_{i_1\cdots i_m}V_{i_1\cdots i_m}^{-1}\\
= & V_{i_1\cdots i_m} T_{i_1\cdots i_m}^{-1} W_{i_{1+2m}\cdots i_{l+m}} U_{i_{1+2m}\cdots i_{l+m}}^{-1} \left(R_{i_1\cdots i_l}  R_{i_{1+m}\cdots i_{l+m}} R_{i_1\cdots i_l}\right) \\& \qquad \qquad U_{i_{1+2m}\cdots i_{l+m}} W_{i_{1+2m}\cdots i_{l+m}}^{-1} T_{i_1\cdots i_m} V_{i_1\cdots i_m}^{-1}.
\end{align*}
    Similarly, the r.h.s. of the gYBE simplifies to 
    \begin{align*}
    & \tilde{R}^{(m,1)}_{i_{1+m}\cdots i_{l+m}} \tilde{R}^{(m,1)}_{i_1\cdots i_l} \tilde{R}^{(m,1)}_{i_{1+m}\cdots i_{l+m}}\\
        = & V_{i_1\cdots i_m} T_{i_1\cdots i_m}^{-1} W_{i_{1+2m}\cdots i_{l+m}} U_{i_{1+2m}\cdots i_{l+m}}^{-1} \left(R_{i_{1+m}\cdots i_{l+m}} R_{i_1\cdots i_l} R_{i_{1+m}\cdots i_{l+m}}\right) \\& \qquad \qquad U_{i_{1+2m}\cdots i_{l+m}} W_{i_{1+2m}\cdots i_{l+m}}^{-1} T_{i_1\cdots i_m} V_{i_1\cdots i_m}^{-1}.
    \end{align*}
    This completes the proof. 
\end{proof}
\begin{remark}
In particular, suppose that the operators $U^{(r)}_{i_{1+r}\cdots i_l}$ and $W^{(r)}_{i_{1+r}\cdots i_l}$ admit the decompositions
\begin{eqnarray*}
U^{(r)}_{i_{1+r}\cdots i_l} = \widetilde{U}^{(r)}_{i_{1+r}\cdots i_{l-r}} V^{(r)}_{i_{l-r+1}\cdots i_l}, \\
W^{(r)}_{i_{1+r}\cdots i_l} = \widetilde{W}^{(r)}_{i_{1+r}\cdots i_{l-r}}
T^{(r)}_{i_{l-r+1}\cdots i_l},
\end{eqnarray*}
where $\widetilde{U}^{(r)}_{i_{1+r}\cdots i_{l-r}}$ and $\widetilde{W}^{(r)}_{i_{1+r}\cdots i_{l-r}}$ are invertible entangled
operators acting on the intermediate tensor block $i_{1+r},\ldots,i_{l-r}$.
A special case of the above theorem is obtained by setting these
auxiliary operators to the identity,
\begin{eqnarray*}
\widetilde{U}^{(r)}_{i_{1+r}\cdots i_{l-r}}= \widetilde{W}^{(r)}_{i_{1+r}\cdots i_{l-r}} = \mathbf{I}.
\end{eqnarray*}
Therefore, the commutation conditions of the theorem reduce to
\begin{eqnarray*}
\left[ R_{i_1\cdots i_l}, V^{(r)}_{i_1\cdots i_r}V^{(r)}_{i_{l-r+1}\cdots i_l} \right] =
\left[ R_{i_1\cdots i_l}, T^{(r)}_{i_1\cdots i_r}T^{(r)}_{i_{l-r+1}\cdots i_l} \right]
= 0.
\end{eqnarray*}
Consequently, if $R_{i_1\cdots i_l}$ satisfies the $(d,l,m)$-gYBE, then the transformed $R$-matrix
\begin{eqnarray*}
\tilde{R}^{(r,1)}_{i_1\cdots i_l} = V^{(r)}_{i_1\cdots i_r} T^{(r)}_{i_{l-r+1}\cdots i_l}
R_{i_1\cdots i_l} \left(V^{(r)}_{i_1\cdots i_r}\right)^{-1} \left(T^{(r)}_{i_{l-r+1}\cdots i_l}\right)^{-1}
\end{eqnarray*}
also satisfies the $(d,l,m)$-gYBE. It is important to notice that for each of the $(d,l,m)$-gYBEs with $m$ chosen from these two sets, a $(d,2r,r)$-gYBE is embedded in the first and last $r=l-m$ indices of $R$.
\end{remark}

\begin{remark}
    The preceding theorem admits a particularly simple specialization when the twisting operators are local. In this case, there is no need to introduce the additional operators $U^{(m)}_{i_{1+m}\cdots i_l}$ and $T^{(m)}_{i_1\cdots i_m}$ explicitly. Instead, it is sufficient to verify the commutation of $R_{i_1\cdots i_l}$ with the global tensor products of the local operators.
    Consider
\begin{equation*}
V^{(m)}_{i_1\cdots i_m} =\bigotimes_{k=1}^{m}v_{i_k}~,\qquad
W^{(m)}_{i_{1+m}\cdots i_l} =\bigotimes_{k=1+m}^{l}w_{i_k},
\end{equation*}
where $v_{i_k},w_{i_k}\in \mathrm{GL}(d)$, $k=1,\ldots,l$, are invertible local operators. Assume that
\begin{equation*}
\left[R_{i_1\cdots i_l},\bigotimes_{k=1}^{l}v_{i_k}\right]=0~,\qquad
\left[R_{i_1\cdots i_l},\bigotimes_{k=1}^{l}w_{i_k}\right]=0.
\end{equation*}
Then the transformed $R$-matrix
\begin{equation*}
\tilde{R}^{(m,1)}_{i_1\cdots i_l} = \left(\bigotimes_{k=1}^{m}v_{i_k}\right)
\left(\bigotimes_{k=1+m}^{l}w_{i_k}\right) R_{i_1\cdots i_l}
\left(\bigotimes_{k=1}^{m}v_{i_k}\right)^{-1}
\left(\bigotimes_{k=1+m}^{l}w_{i_k}\right)^{-1}
\end{equation*}
is again a solution of the $(d,l,m)$-gYBE. Indeed, this is an immediate consequence of the preceding theorem by choosing 
\begin{equation*}
U^{(m)}_{i_{1+m}\cdots i_l} =\bigotimes_{k=1+m}^{l}v_{i_k}~,\qquad
T^{(m)}_{i_1\cdots i_m} =\bigotimes_{k=1}^{m}w_{i_k}.
\end{equation*} 
\end{remark}
For values of $m$ outside the specified sets, twist transformations generated by these entangled operators cannot be constructed. The reason is that the operators
\[
V^{(m)}_{i_{1+m}\cdots i_{l-m}}
  (T^{(m)}_{i_{1+m}\cdots i_{l-m}})^{-1}
\quad\text{and}\quad
T^{(m)}_{i_{1+m}\cdots i_{l-m}}
  (V^{(m)}_{i_{1+m}\cdots i_{l-m}})^{-1}
\]
remain trapped between the $R$-matrices, and in general cannot be moved through them to either side of the gYBE. We next analyze the case in which $m \leq \lfloor \frac{l}{2} \rfloor$ for odd $l$, and $m <  \frac{l}{2}$ for even $l$.
\begin{theorem}[Twist transformations of the gYBE-II-2] \label{thm:twist-gYBE-5}
    Let $R_{1-l}$ be a solution to the $(d,l,m)$-gYBE and acting on $l$ tensor indices. In this case, the admissible values of $m$ are selected as follows:
\begin{eqnarray*}
    && m\in \{1,\ldots,p\}, \quad \text{if } l \text{ is odd}~ (l=2p+1),\nonumber\\
    && m\in \{1,\ldots,p-1\}, \quad \text{if } l \text{ is even}~(l=2p).
\end{eqnarray*}
Let $V_{i_1\cdots i_m}$ and $W_{i_{(l-m)+1}\cdots i_l}$ be invertible operators acting respectively on the left and right boundary blocks consisting of $m$ consecutive indices. Then the transformed $R$-matrix
\begin{equation*}
    \tilde{R}_{i_1\cdots i_l}= V_{i_1\cdots i_m}W_{i_{(l-m)+1}\cdots i_l}~ R_{i_1\cdots i_l}~ V_{i_1\cdots i_m}^{-1}W_{i_{(l-m)+1}\cdots i_l}^{-1},
\end{equation*}
defines a valid twist transformation provided that
\begin{equation*}
    \left[R_{i_1\cdots i_l},  V_{i_1\cdots i_m}V_{i_{(l-m)+1}\cdots i_l}\right] = 0~;~\left[R_{i_1\cdots i_l},  W_{i_1\cdots i_m}W_{i_{(l-m)+1}\cdots i_l}\right]=0.
\end{equation*}
The operators $V_{i_1\cdots i_m}$ and $W_{i_{(l-m)+1}\cdots i_l}$ act exclusively on two boundary blocks of equal size $m$, located at the left and right ends of the $R$-matrices. The central block containing $(l-2m)$ indices is left untouched by the transformation. Indeed, in the regime $(l-3m)>0$, one observes that the $(d,l-2m,m)$-gYBE structure is directly embedded within the central block of the larger $(d,l,m)$-gYBE and remains completely unaffected by the twist elements. Consequently, the twist is localized at the boundaries while preserving the structure of the intermediate sector.
\end{theorem}
\begin{proof}
We begin with the twisted version of the gYBE, defined by the twist elements $V$ and $W$. The l.h.s. then simplifies as follows: 
    \begin{align*}
    &\tilde{R}_{i_1\cdots i_l} \tilde{R}_{i_{1+m}\cdots i_{l+m}} \tilde{R}_{i_1\cdots i_l}\\
    =& V_{i_1\cdots i_m}W_{i_{(l-m)+1}\cdots i_l}R_{i_1\cdots i_l}V_{i_1\cdots i_m}^{-1}W_{i_{(l-m)+1}\cdots i_l}^{-1} V_{i_{1+m}\cdots i_{2m}}W_{i_{l+1}\cdots i_{l+m}}R_{i_{1+m}\cdots i_{l+m}}\\
    & \qquad V_{i_{1+m}\cdots i_{2m}}^{-1}W_{i_{l+1}\cdots i_{l+m}}^{-1}V_{i_1\cdots i_m}W_{i_{(l-m)+1}\cdots i_l}R_{i_1\cdots i_l}V_{i_1\cdots i_m}^{-1}W_{i_{(l-m)+1}\cdots i_l}^{-1}\\
    = & W_{i_{l+1}\cdots i_{l+m}}V_{i_1\cdots i_m}W_{i_1\cdots i_m}^{-1} R_{i_1\cdots i_l}W_{i_1\cdots i_m}V_{i_{l+1}\cdots i_{l+m}}^{-1} R_{i_{1+m}\cdots i_{l+m}}V_{i_{l+1}\cdots i_{l+m}}\\
    & \qquad \qquad W_{i_1\cdots i_m}^{-1}R_{i_1\cdots i_l}W_{i_1\cdots i_m}V_{i_1\cdots i_m}^{-1}W_{i_{l+1}\cdots i_{l+m}}^{-1}\\
    =& W_{i_{l+1}\cdots i_{l+m}} V_{i_{l+1}\cdots i_{l+m}}^{-1}V_{i_1\cdots i_m}W_{i_1\cdots i_m}^{-1} (R_{i_1\cdots i_l}R_{i_{1+m}\cdots i_{l+m}}R_{i_1\cdots i_l})\\
    & \qquad \qquad W_{i_1\cdots i_m}V_{i_1\cdots i_m}^{-1}V_{i_{l+1}\cdots i_{l+m}} W_{i_{l+1}\cdots i_{l+m}}^{-1}.
    \end{align*}
    The r.h.s. will similarly be simplified to 
    \begin{align*}
        &\tilde{R}_{i_{1+m}\cdots i_{l+m}} \tilde{R}_{i_1\cdots i_l} \tilde{R}_{i_{1+m}\cdots i_{l+m}}\\
        =& W_{i_{l+1}\cdots i_{l+m}} V_{i_{l+1}\cdots i_{l+m}}^{-1}V_{i_1\cdots i_m}W_{i_1\cdots i_m}^{-1}(R_{i_{1+m}\cdots i_{l+m}} R_{i_1\cdots i_l} R_{i_{1+m}\cdots i_{l+m}})\\
        & \qquad \qquad W_{i_1\cdots i_m}V_{i_1\cdots i_m}^{-1}V_{i_{l+1}\cdots i_{l+m}} W_{i_{l+1}\cdots i_{l+m}}^{-1}.
    \end{align*}
    Therefore, the twisted $R$-matrix $\tilde{R}$ is also a solution of the gYBE, which completes the proof.
\end{proof}
From the above Theorems we see that $m=\left\lfloor\frac{l}{2}\right\rfloor$ determines the boundary value of $m$ at which the nature of the twist or constrained gauge symmetries change. These Theorems account for twist or constrained gauge symmetries of the $(d,l,m)$-gYBE for all values of $l$ and their allowed $m$ values. 

The fourth and final set of continuous symmetries of the YBE are generated by a set of braid generators, $P_{ij}V_{ij}$, where $P$ is the standard permutation or exchange operator on $\mathbb{C}^d\otimes\mathbb{C}^d$ and $V$ is an invertible operator on the same space. Further restrictions on $V$ lead to a YBE preserving transformation $VRV^{-1}$, when $R$ solves the YBE. These symmetries are known as {\it generalized twists}. We will use the same terminology when addressing these operators. They were first introduced in \cite{G_hmann_1997,Kundu_2001}. Just as the twist symmetries studied earlier, the generalized twist symmetries too make an appearance in the context of the Fermi-Hubbard model as discussed in Chapter 12, Lemma 12 of \cite{esslerfrahmgohmannklumperkorepin2005}. As $V$ is an invertible operator in $\mathrm{GL}(d^2,\mathbb{C}) $, but is constrained, these operators are part of the family of constrained gauge symmetries of the YBE. We will recall the generalized twist symmetry with the following Lemma.
\begin{lemma}[Generalized twist symmetries of the YBE]\label{lem:gen-twist-YBE}
  Let $P_{ij}V_{ij}$ satisfy the YBE. This implies that $ V_{ij}V_{ik}V_{jk} = V_{jk}V_{ik}V_{ij}.$ 
  Furthermore, if $R$ solves the YBE, and satisfies the constraints 
  $\left[R_{ij}, V_{ik}V_{jk}\right] = 0$ and $\left[R_{jk}, V_{ik}V_{ij} \right] = 0$, then $V_{ij}R_{ij}V_{ij}^{-1}$ also solves the YBE.
\end{lemma}
\begin{proof}
    The result follows from a direct computation. Consider the l.h.s of the YBE for the transformed $R$ operator:
    \begin{align*}
        & \tilde{R}_{ij}\tilde{R}_{jk}\tilde{R}_{ij} \\
        = & V_{ij}R_{ij}V_{ij}^{-1}V_{jk}R_{jk}V_{jk}^{-1}V_{ij}R_{ij}V_{ij}^{-1} \\ 
        = & V_{ij}V_{ik}V_{jk}V_{jk}^{-1}V_{ik}^{-1}R_{ij}V_{ij}^{-1}V_{jk}R_{jk}V_{jk}^{-1}V_{ij}R_{ij}V_{ij}^{-1} \\ 
        = & \left(V_{ij}V_{ik}V_{jk}\right)R_{ij}V_{jk}^{-1}V_{ik}^{-1}V_{ij}^{-1}V_{jk}R_{jk}V_{jk}^{-1}V_{ij}R_{ij}V_{ij}^{-1} \\ 
        = & \left(V_{ij}V_{ik}V_{jk}\right)R_{ij}V_{ij}^{-1}V_{ik}^{-1}R_{jk}V_{jk}^{-1}V_{ij}R_{ij}V_{ij}^{-1} \\ 
        = & \left(V_{ij}V_{ik}V_{jk}\right)R_{ij}R_{jk}V_{ij}^{-1}V_{ik}^{-1}V_{jk}^{-1}V_{ij}R_{ij}V_{ij}^{-1} \\ 
        = & \left(V_{ij}V_{ik}V_{jk}\right)R_{ij}R_{jk}V_{jk}^{-1}V_{ik}^{-1}R_{ij}V_{ij}^{-1} \\ 
        = & \left(V_{ij}V_{ik}V_{jk}\right)R_{ij}R_{jk}R_{ij}\left(V_{jk}^{-1}V_{ik}^{-1}V_{ij}^{-1}\right).
    \end{align*}
    In the above computation we have inserted an identity in the first step, made repeated use of the inverse of the YBE and used the relations $$V_{jk}^{-1}V_{ik}^{-1}R_{ij} = R_{ij}V_{jk}^{-1}V_{ik}^{-1}.$$
    All of these relations are obtained from the constraints obeyed by $V$. In a manner similar to the one shown above, the r.h.s of the YBE for the transformed $R$ simplifies to,
    \begin{align*}
         & \tilde{R}_{jk}\tilde{R}_{ij}\tilde{R}_{jk} \\
         = & \left(V_{jk}V_{ik}V_{ij}\right)R_{jk}R_{ij}R_{jk}\left(V_{ij}^{-1}V_{ik}^{-1}V_{jk}^{-1}\right),
    \end{align*}
    showing that $\tilde{R}=VRV^{-1}$ also satisfies the YBE. This completes the proof.
\end{proof}
The generalized twist symmetries stated in Lemma \ref{lem:gen-twist-YBE} may seem obscure in the way it is stated. It becomes clearer when we consider the constraints satisfied by the braid operator $P_{ij}V_{ij}$. Then the constraints between $R$ and $PV$ from Lemma \ref{lem:gen-twist-YBE} reduce to
\begin{equation}
    R_{jk}P_{ij}V_{ij}P_{jk}V_{jk} = P_{ij}V_{ij}P_{jk}V_{jk}R_{ij}.
\end{equation}
This relation shows that the operator $PV$ exchanges the indices of $R$, playing a role similar to that of the standard permutation operator. However, $PV$ is not quite the permutation operator itself, as it does not necessarily square to one. Thus we can view the operator $PV$ as a {\it dressed} permutation operator. Then the generalized twist symmetries can be seen as an extension of the discrete permutation symmetry of the YBE solutions, to be discussed in Section \ref{subsec:d-symmetry-III}.

We will now find analogous generalized twist symmetries for the multi-site gYBEs. Just as in the case of the twist symmetries of the gYBE, here too the answers depend on two sets of $m$ values for both odd and even $l$. These are formulated in the following two Theorems.
\begin{theorem}[Generalized twist symmetries of the gYBE - I]\label{thm:gen-twist-gYBE-I}
 For a given $l$, let $m$ take values in $\lceil l/2\rceil\leq m\leq l-1$. Consider the invertible operator $$V_{\vec{i},\vec{j}}\equiv V_{i_1\cdots i_{l-m};j_{m+1}\cdots j_l},$$ belonging to $\mathrm{GL}(d^{2(l-m)},\mathbb{C})$ such that it satisfies the relation
$$ V_{\vec{i},\vec{j}}V_{\vec{i},\vec{k}}V_{\vec{j},\vec{k}} = V_{\vec{j},\vec{k}}V_{\vec{i},\vec{k}}V_{\vec{i},\vec{j}}, $$
where $\vec{k} = k_{2m+1}\cdots k_{l+m}$. Furthermore we also have the constraints 
\begin{align*}
      \left[R_{\vec{i}\vec{u}\vec{j}}, V_{\vec{i}\vec{k}}V_{\vec{j}\vec{k}}\right] = \left[R_{\vec{j}\vec{w}\vec{k}}, V_{\vec{i}\vec{k}}V_{\vec{i}\vec{j}} \right] = 0.
  \end{align*}
 Here $\vec{u}=u_{l-m+1}\cdots u_m$ and $\vec{w}=w_{l+1}\cdots w_{2m}$ are the indices between $\vec{i}$, $\vec{j}$ and $\vec{j}$, $\vec{k}$ respectively. 
Then $\tilde{R} = VRV^{-1}$ also satisfies the $(d,l,m)$-gYBE.  
\end{theorem}
\begin{proof}
    The result follows from observing that the index configuration of the $(d,l,m)$-gYBE for $m\in\lceil l/2\rceil\leq m\leq l-1$, contains the index structure of the $(d,2(l-m), l-m)$-gYBE. This is indeed true as the integer $l-m$ counts the overlapping indices between two adjacent $R$'s in the $(d,l,m)$-gYBE. More explicitly, this can be seen by splitting the $l$ indices of the gYBE solution for two adjacent $R$'s in the following way:
    \begin{align*}
        & \left(1,\cdots, l-m\right)~;~\left(l-m+1,\cdots, m\right)~;~\left( m+1, \cdots, l\right) \\
        & \left(m+1,\cdots, l\right)~;~\left(l+1,\cdots, 2m\right)~;~\left( 2m+1, \cdots, l+m\right).
    \end{align*}
    These are precisely the indices labeled by $\vec{i}$, $\vec{u}$, $\vec{j}$, $\vec{w}$ and $\vec{k}$ respectively.
     The invertible operator $V$ precisely acts on the first and last $l-m$ indices of the gYBE solution $R$. The proof then follows by applying Lemma \ref{lem:gen-twist-YBE} to these indices.
\end{proof}
The multi-site operator $V$ in this Theorem can be made into a braid group generator by multiplying it with the permutation operator 
\begin{equation}
    P_{\vec{i}, \vec{j}} = \prod_{n=1}^{l-m}~P_{n,m+n}.
\end{equation}
Then it is easy to show that the operator $P_{\vec{i}, \vec{j}}V_{\vec{i}, \vec{j}}$ satisfies the $(d,2(l-m),l-m)$-gYBE. With this result, we can see that the constraints between $V$ and $R$ in Theorem \ref{thm:gen-twist-gYBE-I}, can be reinterpreted as the operator $PV$ acting as a generalized permutation operator on the gYBE solution $R$.

At this point a couple of simple examples are useful to understand Theorem \ref{thm:gen-twist-gYBE-I}. If $R$ satisfies the $(d,4,3)$-gYBE
\begin{align}
   & R_{i_1i_2i_3i_4}R_{i_4i_5i_6i_7}R_{i_1i_2i_3i_4}  \nonumber \\
   = & R_{i_4i_5i_6i_7}R_{i_1i_2i_3i_4}R_{i_4i_5i_6i_7},
\end{align}
then $V_{il}R_{ijkl}V_{il}^{-1}$ also satisfies the $(d,4,3)$-gYBE, provided $P_{ij}V_{ij}$ acts like a dressed permutation operator on $R$. For an odd $l$ example, consider $R$ satisfying the $(d,5,3)$-gYBE
\begin{align}
    &  R_{i_1i_2i_3i_4i_5}R_{i_4i_5i_6i_7i_8}R_{i_1i_2i_3i_4i_5}  \nonumber \\
   = & R_{i_4i_5i_6i_7i_8}R_{i_1i_2i_3i_4i_5}R_{i_4i_5i_6i_7i_8},
\end{align}
then the transformation $V_{ij;lm}R_{ijklm}V_{ij;lm}^{-1}$ is a symmetry of the $(d,5,3)$-gYBE provided that $P_{il}P_{jm}V_{ij;lm}$ acts as a dressed permutation on gYBE solution $R$. 

We are now left with finding the generalized twist symmetries for low $m$ values, i.e. for the case when $m\in\{1,\cdots, p\}\left(\{1,\cdots, p-1\}\right)$ corresponding to $l=2p+1(l=2p)$. From Lemma \ref{lem:gen-twist-YBE}, we expect such symmetries to exist for these $m$ values when we are able to embed the indices of the $(d,2s,s)$-gYBE inside the indices of $(d,l,m)$-gYBE as shown in Theorem \ref{thm:gen-twist-gYBE-I}. This is the case when the multi-site $R$ we consider is a generic solution not taking any particular form. By this we mean a $R$ that does not factorize over the indices on which they act. An example of such a factorization is given by,
$$R_{i_1\cdots i_l} \equiv Y_{i_1i_2}M_{i_3\cdots i_l}.$$
Excluding such cases we find that the only low-\(m\) case captured by this embedding construction is for $l=2p+1$ and $m=p$ when $p>1$. This result can be formulated in the following Theorem.
\begin{theorem}[Generalized twist symmetries of the gYBE - II]\label{thm:gen-twist-gYBE-II}
For odd $l=2p+1$, let $m=p>1$. Then the $(d,2p+1,p)$-gYBE solution $R$, transformed by the generalized twist transformation
\begin{equation*}
    R_{i_1\cdots i_l} \rightarrow 
      V_{i_2\cdots i_{p};i_{p+2}\cdots i_{2p}}R_{i_1\cdots i_l}V_{i_2\cdots i_{p};i_{p+2}\cdots i_{2p}}^{-1}
\end{equation*}
also satisfies the $(d,l,m)$-gYBE, provided that \(P_{\vec i,\vec j}V_{\vec i,\vec j}\) and \(V_{\vec i,\vec j}\) satisfy the corresponding braid and commutation conditions of Lemma \ref{lem:gen-twist-YBE} on the embedded \((d,2(p-1),p-1)\)-gYBE.
\end{theorem}
\begin{proof}
    The result follows by identifying the index structure of a $(d,2(p-1),p-1)$-gYBE in the $(d,2p+1,p)$-gYBE. This is seen by splitting the indices of two adjacent $R$'s in $(d,2p+1,p)$-gYBE as follows:
    \begin{align*}
        & 1~;~(2,\cdots, p)~;~p+1~;~(p+2,\cdots, 2p)~;~2p+1 \\
        & p+1~;~(p+2,\cdots, 2p)~;~2p+1~;~(2p+2,\cdots, 3p)~;~3p+1.
    \end{align*}
    Then supporting $V$ on the indices in the parentheses and using Lemma \ref{lem:gen-twist-YBE}, yields the desired result.
\end{proof}
Now we provide a couple of examples to illustrate Theorem \ref{thm:gen-twist-gYBE-II}. Let $R$ solve the $(d,7,3)$-gYBE
\begin{align}\label{eq:7-3-gYBE}
    & R_{i_1i_2i_3i_4i_5i_6i_7}R_{i_4i_5i_6i_7i_8i_9i_{10}}R_{i_1i_2i_3i_4i_5i_6i_7} \nonumber \\
    = & R_{i_4i_5i_6i_7i_8i_9i_{10}}R_{i_1i_2i_3i_4i_5i_6i_7}R_{i_4i_5i_6i_7i_8i_9i_{10}}.
\end{align}
Then we see that the indices $(2,3;5,6;8,9)$ coincide with the indices appearing in a $(d,4,2)$-gYBE. Thus the gYBE solution rotated by $V_{i_2i_3;i_5i_6}$ will also satisfy the $(d,7,3)$-gYBE following Lemma \ref{lem:gen-twist-YBE}.

For the next example consider the $(d,9,4)$-gYBE
\begin{align}
    & R_{i_1i_2i_3i_4i_5i_6i_7i_8i_9}R_{i_5i_6i_7i_8i_9i_{10}i_{11}i_{12}i_{13}}R_{i_1i_2i_3i_4i_5i_6i_7i_8i_9} \nonumber \\
    = & R_{i_5i_6i_7i_8i_9i_{10}i_{11}i_{12}i_{13}}R_{i_1i_2i_3i_4i_5i_6i_7i_8i_9}R_{i_5i_6i_7i_8i_9i_{10}i_{11}i_{12}i_{13}}.
\end{align}
In this case the indices of a $(d,6,3)$-gYBE are found among
\[
(2,3,4;6,7,8;10,11,12).
\]
Thus the gYBE solution rotated by $V_{i_2i_3i_4;i_6i_7i_8}$ satisfies the $(d,9,4)$-gYBE.

\begin{remark}
For low values of $m$, Theorems \ref{thm:gen-twist-gYBE-I} and \ref{thm:gen-twist-gYBE-II} show that the generalized-twist construction based on the embedded $(d, 2s, s)$-gYBE does not apply to generic multi-site $R$-matrices, except in the cases identified above. We can nevertheless formulate such generalized twist symmetries for these cases when the multi-site $R$-matrices take special forms. These possibilities are best illustrated with examples. Consider a $(d,3,1)$-gYBE:
\begin{equation*}
    R_{i_1i_2i_3}R_{i_2i_3i_4}R_{i_1i_2i_3} = R_{i_2i_3i_4}R_{i_1i_2i_3}R_{i_2i_3i_4},
\end{equation*}
It is easy to see that this gYBE admits two special solutions constructed from standard YBE solutions, $Y$. These are given by 
$$ R_{i_1i_2i_3} \equiv Y_{i_1i_2}~;~ R_{i_1i_2i_3} \equiv Y_{i_2i_3}.$$
Then Lemma \ref{lem:gen-twist-YBE} can be applied for these types of solutions on the indices $(i_1, i_2)$ and $(i_2, i_3)$. The corresponding generalized twist symmetries are given by $V_{i_1i_2}$ and  $V_{i_2i_3}$, respectively.

Next consider the $(d,6,2)$-gYBE
\begin{align*}
    & R_{i_1i_2i_3i_4i_5i_6}R_{i_3i_4i_5i_6i_7i_8}R_{i_1i_2i_3i_4i_5i_6} \nonumber \\
    = & R_{i_3i_4i_5i_6i_7i_8}R_{i_1i_2i_3i_4i_5i_6}R_{i_3i_4i_5i_6i_7i_8}.
\end{align*}
Now there are at least four special solutions of this equation using the standard YBE\footnote{Note that when the indices are disjoint we can form more special solutions. An example in the $(d,6,2)$-gYBE is $Y_{i_1i_3}W_{i_2i_4}$, with $Y$ and $W$ being two different YBE solutions. Thus the set shown here is minimal.}:
$$ R_{i_1\cdots i_6}\in \left\{Y_{i_1i_3}, Y_{i_2i_4}, Y_{i_3i_5}, Y_{i_4i_6}\right\}.$$
The Lemma \ref{lem:gen-twist-YBE} can then be applied on precisely these sets of indices and thus shows that generalized twist symmetries exist for these special classes of low $m$ solutions.

These examples can be generalized to an arbitrary $(d,l,m)$-gYBE. To see this we observe that we can embed the indices of the standard YBE or the $(d,2,1)$-gYBE into the $(d,l,m)$-gYBE in the following manner. In a $(d,l,m)$-gYBE, the standard YBE sits in the indices $j$, $j+m$ and $j+2m$, for $j\in\{1,\cdots,l-m\}$. Thus applying Lemma \ref{lem:gen-twist-YBE} on these indices helps us to construct the generalized twist symmetry in each of these cases. 

It should be noted that the above construction holds for all values of $l$ and $m$, including the higher values of $m$. However, in the latter cases these generalized twist symmetries are special choices of those obtained in Theorems \ref{thm:gen-twist-gYBE-I} and \ref{thm:gen-twist-gYBE-II}.
\end{remark}

%------------------------------------------------
\section{Discrete Symmetries}
\label{sec:d-symmetries}
%------------------------------------------------
The standard YBE has three types of discrete symmetries. We will recall each one of them separately, contrasting them with the analogous discrete symmetries for the gYBE. Out of the three, two of them generalize to the gYBE case in a straightforward manner. We will look at them first before moving on to the third discrete symmetry that is significantly different in the gYBE case.

%------------------------------------------------
\subsection{Discrete Symmetry - I}
\label{subsec:d-symmetry-I}
%------------------------------------------------
The first discrete symmetry is effected by the transformation,
\begin{equation}\label{eq:discrete-1}
    R \rightarrow R^T,~~T\equiv\textrm{transpose}.
\end{equation}
This is valid for any $(d,l,m)$-gYBE as seen by transposing \eqref{eq:gYBE}.

%------------------------------------------------
\subsection{Discrete Symmetry - II}
\label{subsec:d-symmetry-II}
%------------------------------------------------
The second set of discrete symmetries corresponds to a relabeling of the basis elements of the local Hilbert space, $\mathbb{C}^d$. This symmetry depends on $d$ or the representation chosen for $R$. Let the space $\mathbb{C}^d$ be spanned by the $d$ basis vectors, $\{\ket{\alpha}~|~ \alpha\in(1,\cdots,d) \}$. Consider $d-1$ transposition operators permuting these basis vectors,
\begin{equation}
    P_{\alpha,\alpha+1} = \ket{\alpha}\bra{\alpha+1} + \ket{\alpha+1}\bra{\alpha} + \sum\limits_{\beta\neq\alpha,\alpha+1}~\ket{\beta}\bra{\beta},
\end{equation}
with $\alpha\in(1,\cdots,d-1)$. They generate the permutation group $S_d$.
Then the following transformed multi-site $R$-matrices,
\begin{equation}\label{eq:discrete-2}
    R_{j_1\cdots j_l} \rightarrow \left(\prod\limits_{k=1}^l\left(P_{\alpha,\alpha+1}\right)_{j_k} \right)R_{j_1\cdots j_l}\left(\prod\limits_{k=1}^l\left(P_{\alpha,\alpha+1}\right)_{j_k} \right)~;~\alpha\in\{1,\cdots, d-1\},
\end{equation}
also satisfy the $(d,l,m)$-gYBE.

Notice that the operator in this symmetry transformation is completely factorized over the $l$-fold tensor product and thus is similar to the transformation shown in \eqref{eq:gldc-gauge-symmetry-gYBE}. However, in the case of the basis change, the symmetry transformation is finite and discrete, unlike the similarity transform in \eqref{eq:gldc-gauge-symmetry-gYBE}, which is continuous. Thus we do not view the discrete symmetry of \eqref{eq:discrete-2} as a special case of the continuous similarity transform, \eqref{eq:gldc-gauge-symmetry-gYBE}.

%------------------------------------------------
\subsection{Discrete Symmetry - III}
\label{subsec:d-symmetry-III}
%------------------------------------------------
The third set of discrete symmetries is associated with the rearranging of the indices in the $R$-matrix, such that the transformed operator continues to obey a YBE relation with its indices permuted appropriately. The 2-site $R$-matrix satisfying the standard YBE does not have enough room in its support to accommodate a large number of such permutation symmetries, other than a simple exchange of its two indices. On the other hand, the enlarged support of the multi-site $R$-matrices offers many possibilities for such permutation symmetries in the gYBE. However, this is still heavily constrained by the values of $l$ and $m$ as we will show below. Thus we expect an enhancement of permutation symmetries for the multi-site $R$-matrix. Before addressing this we will recall the permutation symmetry of the standard 2-site $R$-matrix. This will act as a guide in identifying the permutation symmetries of the gYBE.

Consider the following exchange transformation on the 2-site $R$-matrix,
\begin{equation}\label{eq:2-site-permutation}
    R_{ij}\rightarrow\tilde{R}_{ij} = P_{ij}R_{ij}P_{ij},
\end{equation}
with $P$ being the standard permutation operator on $\mathbb{C}^d\otimes\mathbb{C}^d$. Then we can show that the two sides of the braid relation for $\tilde{R}$, reduce to  
\begin{align*}
  \tilde{R}_{ij}\tilde{R}_{jk}\tilde{R}_{ij} & = R_{ji}R_{kj}R_{ji} \\ & = P_{ik}\left(R_{jk}R_{ij}R_{jk}\right)P_{ik}, 
  \\ \tilde{R}_{jk}\tilde{R}_{ij}\tilde{R}_{jk} & = R_{kj}R_{ji}R_{kj} \\ & = P_{ik}\left(R_{ij}R_{jk}R_{ij}\right)P_{ik},
\end{align*}
confirming the fact that they are also braid operators. This implies that the 2-site $R$ operator enjoys an Abelian $\mathbb{Z}_2$ permutation symmetry on its two indices. This transformation can also be viewed as a {\it spatial parity symmetry} of the solution of the YBE. 

We will now find analogous permutation symmetries of the multi-site $R$ operator satisfying the $(d,l,m)$-gYBE. We would naively expect that this operator has a full $S_l$ symmetry, obtained by permuting its $l$ indices. However, the number of permutation symmetries is much smaller than this as we shall now show.

We begin by clarifying what constitutes a permutation symmetry of the $(d,l,m)$-gYBE. This definition will also help us set the criteria required of such permutation symmetries. For a general $(d,l,m)$-gYBE, two adjacent $R$'s with supports on the indices, $\left\{1,\cdots, l\right\}$ and $\left\{m+1,\cdots, m+l\right\}$, overlap on $l-m$ indices given by the set $$\left\{m+1,\cdots, l\right\}.$$ Under a permutation transformation $\pi$, the transformed $R$ is given by 
\begin{eqnarray*}
    \tilde{R}_{i_1\cdots i_l} = R_{i_{\pi(1)},\cdots, i_{\pi(l)}}.
\end{eqnarray*}
For the permutation $\pi$ to be a symmetry, we require it to preserve the overlap pattern of the $(d,l,m)$-gYBE. We seek permutations for which the same local rearrangement can be applied
consistently to both copies of $R$. The essential point is that a tensor factor in the overlap has two local descriptions, one in each copy of $R$, and these two descriptions must assign the same image to that tensor factor.

Consider the physical tensor position $m+a$, where
$1\leq a\leq l-m$. In the first copy $R_{i_1,\ldots,i_l}$, this tensor factor occupies local position $m+a$, whereas in the second copy
$R_{i_{m+1},\ldots,i_{m+l}}$ it occupies local position $a$. Suppose first that a global permutation $\sigma\in S_{l+m}$ preserves the two supports separately. Its restrictions to the two copies must then satisfy
\begin{equation}
 \sigma(t)=\pi(t),
 \qquad 1\leq t\leq l,
 \label{eq:sigma-first-support}
\end{equation}
and
\begin{equation}
 \sigma(m+t)=m+\pi(t),
 \qquad 1\leq t\leq l.
 \label{eq:sigma-second-support}
\end{equation}
Applying \eqref{eq:sigma-first-support} and
\eqref{eq:sigma-second-support} to the same physical position $m+a$ gives
\begin{equation}
 \sigma(m+a)=\pi(m+a),
 \qquad
 \sigma(m+a)=m+\pi(a),
\end{equation}
respectively. Consistency therefore requires
\begin{equation}
 \pi(a+m)=\pi(a)+m,
 \qquad 1\leq a\leq l-m.
 \label{eq:positive-overlap-condition}
\end{equation}
This is the overlap-preserving condition. Conversely, if
\eqref{eq:positive-overlap-condition} holds, the prescriptions
\eqref{eq:sigma-first-support} and \eqref{eq:sigma-second-support} agree on $\left\{m+1,\cdots, l\right\}$ and hence define a consistent global relabeling of the tensor
factors.

The structure imposed by \eqref{eq:positive-overlap-condition} becomes
transparent upon introducing the partially defined shift
\begin{equation}
 T(a)=a+m,
 \qquad 1\leq a\leq l-m.
 \label{eq:partial-shift}
\end{equation}
The overlap condition is precisely
\begin{equation}
 \pi\bigl(T(a)\bigr)=T\bigl(\pi(a)\bigr),
\end{equation}
so that $\pi$ commutes with the shift $T$ wherever $T$ is defined. Notice
also that \eqref{eq:positive-overlap-condition} forces
$\pi(a)\leq l-m$ whenever $a\leq l-m$. Since $\pi$ is bijective, it follows
that $\pi$ preserves the domain of $T$. Thus $\pi$ is an automorphism of the
directed graph whose arrows are
\begin{equation}
 a\longrightarrow a+m,
 \qquad 1\leq a\leq l-m.
 \label{eq:shift-arrows}
\end{equation}

The connected components of this directed graph are finite chains. Their initial vertices are precisely $1,2,\ldots,m$, and the chain beginning at $r\in\{1,2,\ldots,m\}$ is
\begin{equation}
 C_r= \{r,r+m,r+2m,\ldots\}\cap\{1,2,\ldots,l\}.
 \label{eq:m-chains}
\end{equation}
Thus there are precisely $m$ such chains that can be accommodated into the set $\left\{1,\cdots, l\right\}$. To find their lengths we notice that one of the chains should end on the index $l$. This gives us 
\begin{equation}
 l=qm+s, \qquad 0\leq s<m.
 \label{eq:euclidean-division}
\end{equation}
Then the number of vertices in $C_r$ or the length of the chain $C_r$ can be computed using the formula
\begin{equation}
 L_r = \left\lfloor\frac{l-r}{m}\right\rfloor+1=\begin{cases}
 q+1, & 1\leq r\leq s,\\[2mm]
 q,   & s<r\leq m.
 \end{cases}
 \label{eq:chain-lengths}
\end{equation}
This gives us $s$ chains of length $q+1$ and $m-s$ chains of length $q$. Because $\pi$ preserves the arrows in \eqref{eq:shift-arrows}, it maps an entire chain to another chain, preserves the order along that chain, and can only map two chains into one another when they have the same length. The long chains may be permuted arbitrarily among themselves, as may the short chains.

If $\phi$ denotes such a permutation of the chain labels, then
the corresponding local permutation is completely determined by
\begin{equation}
 \pi_\phi(r+km)=\phi(r)+km,
 \qquad 0\leq k<L_r,
 \label{eq:chain-permutation-formula}
\end{equation}
where $L_{\phi(r)}=L_r$. Thus the support-preserving transformations arise from independent permutations of the two classes of equal-length chains.

There is a second possibility in preserving the overlap structure of the $(d,l,m)$-gYBE. A global relabeling may exchange the two
supports, rather than preserve them separately. In this case its restrictions have the form
\begin{equation}
 \sigma(t)=m+\pi(t), \qquad \sigma(m+t)=\pi(t).
 \label{eq:support-exchange-restrictions}
\end{equation}
Consistency on the overlap now requires
\begin{equation}
 m+\pi(m+a)=\pi(a),
\end{equation}
or equivalently
\begin{equation}
 \pi(a+m)=\pi(a)-m, \qquad 1\leq a\leq l-m.
 \label{eq:negative-overlap-condition}
\end{equation}
Equation~\eqref{eq:negative-overlap-condition} reverses the direction of the shift arrows. Consequently, a chain must again be mapped to a chain of the same length, but its order is now reversed. The general form is
\begin{equation}
 \pi(r+km) = \phi(r)+(L_r-1-k)m, \qquad 0\leq k<L_r,
 \label{eq:reversed-chain-formula}
\end{equation}
where $\phi$ permutes chains of equal length. In particular, taking
$\phi$ to be the identity gives the simultaneous internal reversal
\begin{equation}
 J(r+km)=r+(L_r-1-k)m.
 \label{eq:simultaneous-chain-reversal}
\end{equation}
The reversal must be simultaneous: the two full supports are either preserved or exchanged, so one cannot independently choose the orientation of each chain. This accounts for a single additional involution, rather than one independent involution for every chain. The ordinary spatial reversal
\begin{equation}
 \rho(a)=l+1-a
\end{equation}
also satisfies \eqref{eq:negative-overlap-condition}. It may differ from $J$ by a permutation of chains of equal length, but belongs to the same support-exchanging class.

The above observations can be captured into the theorem below describing the general permutation symmetry-classification of the $(d,l,m)$-gYBE.

%The role of Euclidean division in this construction should be distinguished from that of $\gcd(l,m)$. The chains in Eq.~\eqref{eq:m-chains} are obtained by repeatedly adding $m$ only while the resulting position remains in the finite interval $\{1,\ldots,l\}$. There is no reduction modulo $l$. Hence these components are finite paths rather than modular cycles, and their two possible lengths are governed by the remainder $s=l\bmod m$. 

\begin{theorem}[Universal permutation symmetries of the $(d,l,m)$-gYBE]
\label{thm:universal-permutation-symmetries}
Let $R$ satisfy the braid-form $(d,l,m)$-gYBE, and write
\[l=qm+s,\qquad 0\leq s<m. \]
For the chains
\[C_r=\{r,r+m,r+2m,\ldots\}\cap\{1,\ldots,l\}, \qquad L_r=|C_r|,\]
let $\phi$ range over the permutations of $\{1,\ldots,m\}$ satisfying
$L_{\phi(r)}=L_r$. The universal permutation symmetries are
precisely
\[\pi^+_{\phi}(r+km)=\phi(r)+km,~~ \pi^-_{\phi}(r+km)=\phi(r)+(L_r-1-k)m, ~~0\leq k<L_r.\]
The first family preserves the two adjacent supports and the second exchanges
them. Consequently,
\[
 G_{\mathrm{perm}}(l,m)
 \cong S_s\times S_{m-s}\times\mathbb{Z}_2,
 \qquad
 |G_{\mathrm{perm}}(l,m)|=2\,s!\,(m-s)!,
\]
where $S_0$ and $S_1$ are understood to be trivial. The ordinary spatial
reversal $\rho(a)=l+1-a$ belongs to the second family.
\end{theorem}

\begin{proof}
By the preceding discussion, a support-preserving permutation satisfies
\[
 \pi(a+m)=\pi(a)+m.
\]
It is therefore an automorphism of the directed shift graph
$a\longrightarrow a+m$. Hence it permutes chains of equal length and
preserves their internal order, giving exactly the permutations
$\pi^+_{\phi}$. Since there are $s$ chains of length $q+1$ and $m-s$ chains
of length $q$, these permutations form
\[
 H\cong S_s\times S_{m-s}.
\]

A support-exchanging permutation instead satisfies
\[
 \pi(a+m)=\pi(a)-m.
\]
It reverses every arrow of the shift graph and therefore gives exactly the
permutations $\pi^-_{\phi}$. Let $J=\pi^-_{\id}$ be simultaneous reversal
within every chain. Every element of the second family is uniquely of the
form $\pi^+_{\phi}J$. Moreover, $J^2=\id$, and $J$ commutes with $H$ because
the elements of $H$ only exchange chains of equal length. Thus
\[
 G_{\mathrm{perm}}(l,m)
 =H\times\langle J\rangle
 \cong S_s\times S_{m-s}\times\mathbb{Z}_2.
\]
The order formula follows immediately. Finally,
\[
 \rho(a+m)=\rho(a)-m,
\]
so the ordinary spatial reversal belongs to the support-exchanging family.
\end{proof}

We now illustrate the chain construction and the resulting permutation groups through several examples. Throughout, $P_{ij}$ denotes the permutation operator that exchanges the $i$-th and $j$-th local tensor positions, and $\mathbf{1}$ denotes the identity permutation. We first consider the balanced case $l=2m$, followed by two examples with $m<l/2$ and two with $m>l/2$.

\subsubsection*{\texorpdfstring{The balanced $(d,6,3)$-gYBE}
{The balanced (d,6,3)-gYBE}}
Since $6=2\cdot3,$ the three chains all have length two:
\[ C_1=\{1,4\},\qquad C_2=\{2,5\},\qquad C_3=\{3,6\}.\]
They may therefore be permuted arbitrarily. Let
\[ A=P_{12}P_{45}, \qquad B=P_{23}P_{56}.\]
The operator $A$ exchanges $C_1$ and $C_2$, while $B$ exchanges $C_2$ and $C_3$. The support-preserving subgroup is
\[ H_{6,3} =\langle A,B\rangle =\{\mathbf{1},A,B,ABA,AB,BA\} \cong S_3,
\]
where
\[ABA=P_{13}P_{46},\]
\[AB=P_{12}P_{23}P_{45}P_{56}, \qquad BA=P_{23}P_{12}P_{56}P_{45}.
\]
Simultaneous reversal within the three chains is $J=P_{14}P_{25}P_{36}$. Thus all twelve permutation symmetries are
\[G_{\mathrm{perm}}(6,3)= \{hJ^\epsilon\mid h\in H_{6,3},\ \epsilon\in\{0,1\}\} \cong S_3\times\mathbb{Z}_2.\]
The ordinary spatial reversal is $\rho=P_{16}P_{25}P_{34}=(P_{13}P_{46})J$. This also demonstrates that simultaneous internal chain reversal and ordinary spatial reversal need not be represented by the same permutation, although they belong to the same support-exchanging family.

\subsubsection*{\texorpdfstring{Small $m$ and even $l$: the $(d,8,3)$-gYBE}
{Small m and even l: the (d,8,3)-gYBE}}
Here $8=2\cdot3+2$, and the chains are
\[
 C_1=\{1,4,7\},\qquad
 C_2=\{2,5,8\},\qquad
 C_3=\{3,6\}.
\]
The two length-three chains may be exchanged, while the length-two chain cannot be mixed with them. The nontrivial chain exchange is $\tau=P_{12}P_{45}P_{78}$ and simultaneous reversal within all three chains is $J=P_{17}P_{28}P_{36}$. Consequently, all four symmetries are
\[
 G_{\mathrm{perm}}(8,3)
 =
 \{\mathbf{1},\tau,J,\tau J\}
 \cong S_2\times\mathbb{Z}_2.
\]
The ordinary spatial reversal is $\rho=P_{18}P_{27}P_{36}P_{45}=\tau J$. Thus ordinary spatial reversal is not the only nontrivial permutation symmetry of the $(d,8,3)$-gYBE.

\subsubsection*{\texorpdfstring{Small $m$ and odd $l$: the $(d,5,2)$-gYBE}
{Small m and odd l: the (d,5,2)-gYBE}}
In this case, $5=2\cdot2+1$,
\[
 C_1=\{1,3,5\},\qquad
 C_2=\{2,4\}.
\]
The chains have different lengths, so there are no nontrivial chain exchanges.
Simultaneous chain reversal is $J=P_{15}P_{24}$. It coincides with ordinary spatial reversal, and therefore
\[
 G_{\mathrm{perm}}(5,2)
 =
 \{\mathbf{1},P_{15}P_{24}\}
 \cong\mathbb{Z}_2.
\]
This example shows that the support-exchanging $\mathbb{Z}_2$ symmetry remains even when no two chains can be exchanged.

\subsubsection*{\texorpdfstring{Large $m$ and odd $l$: the $(d,7,5)$-gYBE}
{Large m and odd l: the (d,7,5)-gYBE}}
Here $7=1\cdot5+2$. There are two chains of length two,
\[
 C_1=\{1,6\},\qquad C_2=\{2,7\},
\]
and three singleton chains,
\[
 C_3=\{3\},\qquad C_4=\{4\},\qquad C_5=\{5\}.
\]
The two length-two chains are exchanged by $A=P_{12}P_{67}$.
The permutations of the three singleton chains form
\[\mathcal{S}_{345}=\{\mathbf{1}, P_{34}, P_{45}, P_{35}, P_{34}P_{45},
 P_{45}P_{34}\}\cong S_3.\]
Simultaneous chain reversal acts nontrivially only on the two length-two chains and is given by $J=P_{16}P_{27}$. All twenty-four symmetries are therefore
\[
 G_{\mathrm{perm}}(7,5)
 =
 \{
 A^\alpha S J^\epsilon
 \mid
 \alpha,\epsilon\in\{0,1\},\
 S\in\mathcal{S}_{345}
 \}
 \cong S_2\times S_3\times\mathbb{Z}_2.
\]
The ordinary spatial reversal is $\rho=P_{17}P_{26}P_{35}
      =A\,P_{35}\,J$.

\subsubsection*{\texorpdfstring{Large $m$ and even $l$: the $(d,4,3)$-gYBE}
{Large m and even l: the (d,4,3)-gYBE}}
Finally, $4=1\cdot3+1$, and the chains are
\[
 C_1=\{1,4\},\qquad
 C_2=\{2\},\qquad
 C_3=\{3\}.
\]
The two singleton chains may be exchanged by $\tau=P_{23}$, while internal chain reversal is $J=P_{14}$. Hence all four permutation symmetries are
\[
 G_{\mathrm{perm}}(4,3)= \{\mathbf{1},P_{23},P_{14},P_{14}P_{23}\}
 \cong S_2\times\mathbb{Z}_2.
\]
The ordinary spatial reversal is $\rho=P_{14}P_{23}=\tau J$.

These examples display the different possibilities contained in the general
classification. For small $m$, either nontrivial exchanges of equal-length
chains are possible or only the reversal symmetry survives. For large $m$,
the singleton chains supply the additional symmetric-group factor.

This completes our analysis of the permutation symmetries of the $(d,l,m)$-gYBE. We end this discussion with the following remark.
The permutation symmetries of the solutions of the YBE can be seen from a different perspective. Instead of permuting the indices on the YBE solution $R$, we can instead shuffle the indices $i$, $j$, $k$ in the YBE, and ask the question if the transformed YBE can be obtained by permuting indices on the 2-site $R$ operator. For example, consider permuting the indices $i$ and $j$ on both sides of the YBE:
\begin{align}
   & P_{ij}\left(R_{ij}R_{jk}R_{ij}\right)P_{ij} \\
  \nonumber & = \tilde{R}_{ij}P_{ij}P_{jk}\tilde{R}_{jk}P_{jk}P_{ij}\tilde{R}_{ij} \\
 \nonumber & = \tilde{R}_{ij}\tilde{R}_{ki}\tilde{R}_{ij} \\
\nonumber & = P_{jk}P_{ij}\left(\tilde{R}_{jk}\tilde{R}_{ij}\tilde{R}_{jk} \right) P_{ij}P_{jk}, \\
  & P_{ij}\left(R_{jk}R_{ij}R_{jk}\right)P_{ij} \\
\nonumber & =\left(P_{ij}P_{jk}\tilde{R}_{jk}P_{jk}P_{ij}\right)\tilde{R}_{ij}\left(P_{ij}P_{jk}\tilde{R}_{jk}P_{jk}P_{ij}\right) \\
\nonumber & = \tilde{R}_{ki}\tilde{R}_{ij}\tilde{R}_{ki} \\
\nonumber & = P_{jk}P_{ij}\left(\tilde{R}_{ij}\tilde{R}_{jk}\tilde{R}_{ij} \right) P_{ij}P_{jk}.
\end{align}
A similar computation holds when the indices $j$ and $k$ are interchanged in the YBE. The above computations reveal the following two points. First, the 2-site $R$-matrix, under the interchange of its two indices, satisfies the braid relation. This is a $\mathbb{Z}_2$ symmetry generated by the standard permutation exchanging the two indices of the $R$-matrix. And second, the standard YBE is covariant under any permutation of the three indices it is supported on. This is realized by the transpositions, $P_{ij}$ and $P_{jk}$ which generate a full $S_3$ symmetry for the YBE. Moreover, every $S_3$ operation on the YBE can be traced back to a $\mathbb{Z}_2$, exchange symmetry on the 2-site $R$-matrix. This leads to the conclusion that for the standard YBE, the two perspectives on the permutation symmetries are entirely equivalent. 

An analogous equivalence for the multi-site $R$-matrix and the associated gYBEs is not always true. From the perspective of the permutation symmetries of the gYBE, the $(d,l,m)$-gYBE has a $S_{l+m}$ permutation symmetry, which implies that an arbitrary permutation of the $l+m$ indices in the gYBE will lead to another gYBE. For this to be compatible with the symmetries of the multi-site $R$, we require that the permutation-transformed gYBE be obtained by one of the $S_l$ operations on the multi-site $R$. However, our analysis of the permutation symmetries of the gYBE has shown this not to be true. This is an important difference between the permutation symmetries of the YBE and the gYBE.

%As a final remark we note that for low values of $m$, we can obtain gYBEs on a lower value of $l$ after partial tracing some of the indices. For example, consider the $(d,3,1)$-gYBE
%\begin{equation}
 %   R_{i_1i_2i_3}R_{i_2i_3i_4}R_{i_1i_2i_3} = R_{i_2i_3i_4}R_{i_1i_2i_3}R_{i_2i_3i_4}.
%\end{equation}
%In this case partial tracing either the index 2 or 3 results in a $(d,2,1)$-gYBE or the standard YBE. This partial tracing can be viewed as establishing an equivalence between a gYBE solution of higher $l$ with a gYBE solution of lower $l$. They correspond to semigroup symmetries of the gYBE. We will study these in the future.

%-------------------------------------------------------------------
\section{\texorpdfstring{Application: A symmetry-reduced ansatz for the constant $(2,4,2)$-gYBE}{Application: A symmetry-reduced ansatz for the constant (2,4,2)-gYBE}}
\label{sec:solutions}
%---------------------------------------------------------------------
The constant $(2,4,2)$-gYBE
\begin{eqnarray}
R_{1234}R_{3456}R_{1234} = R_{3456}R_{1234}R_{3456}
\end{eqnarray}
is non-linear and consists of a set of highly overdetermined equations for the entries of the four-site $R$. If this $R$-matrix is written as a generic $16\times 16$ matrix, one needs to solve for $256$ unknown entries. Explicitly, we write 
\begin{equation}
R = \begin{pmatrix}
r_{1,1} & r_{1,2} & \cdots & r_{1,16} \\
r_{2,1} & r_{2,2} & \cdots & r_{2,16} \\
\vdots  & \vdots  & \ddots & \vdots  \\
r_{16,1} & r_{16,2} & \cdots & r_{16,16}
\end{pmatrix},
\end{equation}
where $r_{i,j}$ denotes the matrix element connecting the basis states $i$ and $j$, with $i,j = 1, \dots, 16$. However, the equation itself produces a very large system of coupled cubic polynomial equations. A direct brute-force solution of this system is therefore impractical. For this reason, we have adopted a step-by-step reduction strategy based first on the discrete symmetries of the equation, then on a guided simplification of the ansatz, and finally on the factorized structure of the resulting polynomial equations. The first simplification comes from imposing the discrete symmetries of the $(2,4,2)$-gYBE. In addition to transposition,
\begin{eqnarray}
R^{T} = R,
\end{eqnarray}
we require invariance under the Pauli-$X$ conjugation and under the relevant
two-site permutation symmetries,
\begin{eqnarray}
R &=& X_{1}X_{2}X_{3}X_{4}\, R\, X_{1}X_{2}X_{3}X_{4}, \nonumber\\
R &=& P_{12}P_{34}\, R\, P_{12}P_{34}, \nonumber\\
R &=& P_{13}P_{24}\, R\, P_{13}P_{24}.
\end{eqnarray}
Here $X_i$ denotes the Pauli-$X$ matrix acting on site $i$, and $P_{ij}$ is the permutation operator that swaps sites $i$ and $j$. When we apply these discrete symmetries, we reduce the number of independent entries from $256$ down to $30$. Essentially, this means that the full $16\times 16$ matrix can be effectively described using a symmetric ansatz with $30$ parameters. The most noticeable effect of these constraints is the symmetry condition $R^{T}=R$, although the remaining discrete symmetries impose additional relations among the matrix elements. 

However, although the reduction due to the discrete symmetry is rather significant, it is not enough. Despite the reduction of the ansatz to $30$ independent parameters, its substitution in the $(2,4,2)$-gYBE leads to an enormous and complicated polynomial system. To make things more manageable, we examined the simplest factorized polynomial equations and decided to simplify the ansatz even further. 
As a first guided truncation, we set the following fourteen parameters to zero:
\begin{eqnarray}
\mathcal{Z}_{0} &=& \{ r_{1,2}, r_{1,4}, r_{1,6}, r_{1,7}, r_{1,8},
 r_{2,6}, r_{2,7}, r_{2,8}, \nonumber\\
&& r_{2,10}, r_{2,11}, r_{2,12}, r_{4,6}, r_{4,7}, r_{6,7} \},
\end{eqnarray}
that is, we imposed
\begin{eqnarray}
z = 0, \quad \forall z \in \mathcal{Z}_{0}.
\end{eqnarray}
Each of the 30 parameters labels an orbit of matrix elements that are interconnected via discrete symmetries. Setting any of these parameters to zero removes from the picture all elements related to this parameter by symmetries. Now, after obtaining this particular pattern of zeros and reducing the resulting polynomial equations to a simpler form, we find an additional factorization pattern in the structure of the equations. Inspection of the resulting polynomial equations reveals that the $r_{2,13}$, $r_{2,14}$ and $r_{2,15}$ parameters repeatedly appear as factors in several equations. The occurrence of such a pattern implies that the above parameters define natural algebraic branches. Therefore we should focus our attention on the corresponding ansatz branch. Motivated by this observation, we further restricted the ansatz to the branch
\begin{eqnarray}
r_{2,13} = r_{2,14}= r_{2,15} = 0.
\end{eqnarray}
Thus the complete set of zero assignments used at this stage is
\begin{eqnarray}
\mathcal{Z}_{1} = \mathcal{Z}_{0} \cup \{ r_{2,13}, r_{2,14}, r_{2,15} \}.
\end{eqnarray}
With this second reduction, the number of remaining independent parameters is substantially decreased; the surviving parameters are
\begin{eqnarray}
\mathcal{V}_{1} &=& \{ r_{1,1}, r_{1,16}, r_{2,2}, r_{2,3}, r_{2,4}, r_{2,5}, r_{2,9}, \nonumber\\
&& r_{4,4}, r_{4,13}, r_{6,6}, r_{6,11}, r_{7,7}, r_{7,10} \}.
\end{eqnarray}
After applying the ansatz, the number of polynomial equations is reduced substantially, from $328$ to $148$. At this stage, we have simplified the polynomial system enough to allow for a detailed algebraic analysis. We begin by using a direct symbolic solution method on the current system, treating all the remaining independent parameters as unknowns and not placing any further restrictions. If we successfully solve the system, we will record the raw solution set as it is, without making any additional simplifications or introducing new branches. If the direct symbolic solver, such as \textsc{Mathematica}, cannot find the solutions, or if the resulting system is still too complex, we will analyze the factorization patterns of the remaining equations to discover more consistent branches. Each reduction we make trims down the number of independent parameters and equations, leading us to smaller polynomial systems where we can apply the same approach over and over.

For the above branch $\mathcal{Z}_{1}$, the reduced ansatz yields a polynomial system that admits $728$ symbolic solution branches. For 190 of these branches, the determinant of \(R\) is not identically zero; these branches therefore contain generically invertible \(R\)-matrices. For the remaining 538 branches, the determinant vanishes identically. Because the gYBE is homogeneous and invariant under the continuous symmetries described below, these branches do not consist of isolated matrices but of parameterized families. In each non-trivial branch, at least the overall scale remains free. An explicit example of a non-trivial invertible branch is given by,
\begin{equation*}
\resizebox{\linewidth}{!}{%
$
\begin{pmatrix}
 0 & 0 & 0 & 0 & 0 & 0 & 0 & 0 & 0 & 0 & 0 & 0 & 0 & 0 & 0 & r_{1,16} \\
 0 & 0 & 0 & 0 & r_{2,5} & 0 & 0 & 0 & r_{2,9} & 0 & 0 & 0 & 0 & 0 & 0 & 0 \\
 0 & 0 & 0 & 0 & r_{2,9} & 0 & 0 & 0 & r_{2,5} & 0 & 0 & 0 & 0 & 0 & 0 & 0 \\
 0 & 0 & 0 & r_{4,4} & 0 & 0 & 0 & 0 & 0 & 0 & 0 & 0 & 0 & 0 & 0 & 0 \\
 0 & r_{2,5} & r_{2,9} & 0 & 0 & 0 & 0 & 0 & 0 & 0 & 0 & 0 & 0 & 0 & 0 & 0 \\
 0 & 0 & 0 & 0 & 0 & r_{6,6} & 0 & 0 & 0 & 0 & r_{7,7} & 0 & 0 & 0 & 0 & 0 \\
 0 & 0 & 0 & 0 & 0 & 0 & r_{7,7} & 0 & 0 & r_{6,6} & 0 & 0 & 0 & 0 & 0 & 0 \\
 0 & 0 & 0 & 0 & 0 & 0 & 0 & 0 & 0 & 0 & 0 & 0 & 0 & r_{2,5} & r_{2,9} & 0 \\
 0 & r_{2,9} & r_{2,5} & 0 & 0 & 0 & 0 & 0 & 0 & 0 & 0 & 0 & 0 & 0 & 0 & 0 \\
 0 & 0 & 0 & 0 & 0 & 0 & r_{6,6} & 0 & 0 & r_{7,7} & 0 & 0 & 0 & 0 & 0 & 0 \\
 0 & 0 & 0 & 0 & 0 & r_{7,7} & 0 & 0 & 0 & 0 & r_{6,6} & 0 & 0 & 0 & 0 & 0 \\
 0 & 0 & 0 & 0 & 0 & 0 & 0 & 0 & 0 & 0 & 0 & 0 & 0 & r_{2,9} & r_{2,5} & 0 \\
 0 & 0 & 0 & 0 & 0 & 0 & 0 & 0 & 0 & 0 & 0 & 0 & r_{4,4} & 0 & 0 & 0 \\
 0 & 0 & 0 & 0 & 0 & 0 & 0 & r_{2,5} & 0 & 0 & 0 & r_{2,9} & 0 & 0 & 0 & 0 \\
 0 & 0 & 0 & 0 & 0 & 0 & 0 & r_{2,9} & 0 & 0 & 0 & r_{2,5} & 0 & 0 & 0 & 0 \\
 r_{1,16} & 0 & 0 & 0 & 0 & 0 & 0 & 0 & 0 & 0 & 0 & 0 & 0 & 0 & 0 & 0 \\
\end{pmatrix}.
$}
\end{equation*}
Direct substitution shows that this matrix satisfies the \((2,4,2)\)-gYBE identically for arbitrary values of its six parameters. This $R$ is invertible for generic parameter values; explicitly, whenever
$r_{1,16}r_{4,4} \bigl(r_{2,5}^{2}-r_{2,9}^{2}\bigr) \bigl(r_{6,6}^{2}-r_{7,7}^{2}\bigr)\neq 0.$ Furthermore we can check that this $(2,4,2)$ $R$-matrix is not obtained from a product of the $(2,2,1)$ Yang-Baxter solutions. To do this we need to ensure that the following factorizations, $R=Y_{13}\widetilde{Y}_{24}$ or
$R=Y_{14}Y_{23}$, are not possible\footnote{Note that both these decompositions satisfy the $(d,4,2)$-gYBE when $Y$ and $\tilde{Y}$ satisfy the standard YBE or the $(d,2,1)$-gYBE.}. Now let $\mathcal{R}_{A\mid B}$ denote operator realignment with respect to the
bipartition $A\mid B$. For the above factorizations to hold we require that $\mathcal{R}_{13\mid24}(R)$ or $\mathcal{R}_{14\mid23}(R)$, respectively, has rank one. However, both realignments contain the minor
\[
\det\begin{pmatrix}r_{1,16}&0\\0&r_{4,4}\end{pmatrix}
=r_{1,16}r_{4,4}\neq0,
\]
where the last inequality follows from the invertibility of $R$. Hence neither
factorization is possible.

Once we have collected the raw solution branches from a specific branch, the next task is to eliminate any duplicates and solutions that are equivalent due to symmetry. It is quite interesting to note that solutions that seem different might actually be related through the continuous symmetries of the gYBE. 
In the present case, the relevant continuous transformations are the block gauge transformations generated by an invertible two-site operator $Q_{ij}, ~Q\in \mathrm{GL}(4,\mathbb{C})$. The resulting continuous entangled conjugation symmetry is
%The first is the local single-site conjugation symmetry,
%\begin{eqnarray}
%R &\longmapsto& \kappa~ (Q_{1}\otimes Q_{2}\otimes Q_{3}\otimes Q_{4})\, R\,
%(Q_{1}\otimes Q_{2}\otimes Q_{3}\otimes Q_{4})^{-1},
%\end{eqnarray}
%where $\kappa \neq 0$ and each $Q_i$ is an invertible single-site operator. This transformation contains the global scaling $R \mapsto \kappa~R$ as the special case $Q_i =\mathbf{1}$. The second is the two-site entangled conjugation symmetry,
\begin{eqnarray}
R &\longmapsto& \kappa~(Q_{12}\otimes Q_{34})\, R\, (Q_{12}\otimes Q_{34})^{-1},
\end{eqnarray}
where $\kappa \neq 0$ and $Q_{12}$ acts on sites $1,2$ and $Q_{34}$ acts on sites $3,4$. This transformation contains the global scaling $R \mapsto \kappa~R$ as the special case $Q_{ij} =\mathbf{1}$.

These continuous symmetries can be used to bring the solutions into simpler
representatives and to identify solutions that are equivalent under gauge
transformations. In addition, one may introduce the partial-trace matrices
\begin{eqnarray}
\rho_{12}(R) = \mathrm{Tr}_{34}\, R, \qquad
\rho_{34}(R) = \mathrm{Tr}_{12}\, R,
\end{eqnarray}
which act on sites $1,2$ and $3,4$, respectively. They transform covariantly under the two-site continuous symmetries,
\begin{eqnarray}
\rho_{12}(R) &\longmapsto& \kappa~ Q_{12}\, \rho_{12}(R)\, Q_{12}^{-1}, \nonumber\\
\rho_{34}(R) &\longmapsto& \kappa~ Q_{34}\, \rho_{34}(R)\, Q_{34}^{-1}.
\end{eqnarray}
These partial-trace matrices are useful in the final stage of the analysis, where they help us create canonical representatives and weed out duplicate solutions that are linked by continuous transformations. Their main function is not to set the initial branch but to support the final classification and identification of genuinely inequivalent solutions, much like the role of trace matrices in Hietarinta's canonical-form algorithm \cite{hietarinta1993-JMP-Long}.

It is worth mentioning that all solutions obtained through the specified branch $\mathcal{Z}_1$ are necessarily symmetric because of the condition $R=R^T$. It is also evident that after imposing $\mathcal{Z}_1$, the symbolic solver returns $728$ algebraic solution branches. This number should not be interpreted as the number of inequivalent $R$-matrices modulo the continuous symmetries. Rather, it only counts the raw branches produced by the solver before removing specializations, and symmetry-equivalent solutions. So a complete classification of the solutions within a single branch is already nontrivial, while other branches may yield additional solutions that are not necessarily related to those we find here \textit{via} continuous or gauge transformations. A systematic classification of all branches is therefore beyond the scope of the present work. Our main focus is more straightforward: we are looking to establish a systematic approach to refine the initial ansatz to uncover genuine four-site solutions of the gYBE, rather than attempting a complete classification of the continuous families of solutions.

%------------------------------------------------
\section{Conclusion}
\label{sec:conclusion}
%------------------------------------------------
The standard YBE or the $(d,2,1)$-gYBE contains a set of highly non-linear and over-determined equations to determine the matrix elements of the solution $R$. One way to simplify this problem is to consider the symmetries of the $R$ operator as they reduce the number of unknowns and thus help in simplifying the ans\"{a}tze used to determine them. This approach has been successfully adopted to find the constant $4\times 4$ YBE solutions in \cite{HIETARINTA-PLA}. Baxterizing these solutions leads to spectral parameter dependent $R$-matrices that are used to construct different integrable spin chains with nearest-neighbor interactions \cite{Maity_2025,Maity_2026}. 

It is natural to extend the above studies to a multi-site generalization of the YBE known as the $(d,l,m)$-gYBE. As these equations are more non-linear and over-determined when compared to the YBE, one would imagine that obtaining solutions for them would become significantly harder. However, in this work we have shown that this in fact need not necessarily be true as the gYBEs have larger symmetry groups that can substantially reduce the space of inequivalent solutions and provide useful symmetry-invariant sectors in which to conduct classification searches. While some of the symmetries of the YBE carry over to the gYBE in a straightforward manner, key differences appear in the case of constrained and unconstrained continuous symmetries and in the case of the discrete permutation symmetries. The latter is formulated in Theorem \ref{thm:universal-permutation-symmetries}, while the results of the former can be found in Theorems \ref{thm:gauge-symmetries}, \ref{thm:twist-gYBE-1}, \ref{thm:twist-gYBE-2}, \ref{thm:twist-gYBE-4}, \ref{thm:twist-gYBE-5}, \ref{thm:gen-twist-gYBE-I} and \ref{thm:gen-twist-gYBE-II}. An explicit example through the $(2,4,2)$-gYBE ansatz demonstrated the usefulness of these symmetries in significantly reducing the number of unknown entries in their matrix ans\"{a}tze.

At this point it is important to emphasize that there could be more symmetries for the gYBEs that are not extensions of analogous symmetries of the YBE. In particular, it would be interesting to consider Drinfeld twist symmetries \cite{kulish1998twistingsolutionsyangbaxterequation} of the gYBE and the underlying quantum group structure. It would also be interesting to look for the symmetries of the vertex models associated to the gYBEs along the lines of \cite{BELLON199187}.

We hope that the results in this work can be used to make judicious choices of ans\"{a}tze for multi-site \(R\)-matrices and thereby facilitate the construction of new solutions with applications to multi-site integrable systems. In particular, it can help us systematically construct integrable models with interactions that go beyond the nearest-neighbor ones. In this context, spectral parameter dependent solutions of the gYBE, based on {\it extraspecial 2-groups}, have recently been used to construct multi-site versions of the Ising models \cite{sinha2026hiddenisingmodelsgeneralized}. In addition to this, some preliminary studies have shown us that the QISM applied on the gYBE can lead to coupled integrable systems. These include examples where a Hermitian integrable system is coupled to a non-Hermitian model \cite{Shibata_2019,Shibata_2019-2}. We will return to these in the near future.

%------------------------------------------------

\subsection*{Acknowledgements}

The authors thank the anonymous OCNMP referee for valuable comments that improved the manuscript.

\label{lastpage}
\end{document}